\documentclass[11pt]{article}

\usepackage{preamble}

\usepackage{cite}
\usepackage{orcidlink}
\usepackage{bbm}
\usepackage{makecell}
\usepackage{mathtools}

\usepackage{pgfplots}
\pgfplotsset{compat=1.18}
\usetikzlibrary{calc}
\usetikzlibrary{spy}
\usepackage{pythonhighlight}

\usepackage{pifont}

\usepackage[dvipsnames]{xcolor}
\usepackage{tikz}
\usetikzlibrary{fit, calc, backgrounds,perspective,3d, arrows.meta}
\usepackage{wrapfig}
\usepackage{enumitem}
\usepackage{booktabs}

\usepackage{mathtools}

\usepackage{hyperref}
\hypersetup{
    colorlinks=true,
    linkcolor=Blue,
    citecolor=Mahogany,
    urlcolor=Mahogany,
}

\usepackage{listings}

\lstdefinestyle{abstractcode}{
  basicstyle=\ttfamily\small,
  keywordstyle=\bfseries,
  commentstyle=\itshape\color{gray},
  columns=fullflexible,
  keepspaces=true,
  frame=single,
  mathescape=true,
  xleftmargin=1em,
  framexleftmargin=1em
}

\definecolor{pykeyword}{RGB}{0,70,140}
\definecolor{pycomment}{RGB}{120,120,120}
\definecolor{pyfunc}{RGB}{90,0,130}   

\definecolor{pycond}{RGB}{180,120,0} 

\lstdefinestyle{pythonabstract}{
    language=Python,
    basicstyle=\ttfamily\small,
    keywordstyle=\bfseries\color{pykeyword},
    commentstyle=\itshape\color{pycomment},
    showstringspaces=false,
    columns=fullflexible,
    keepspaces=true,
    frame=single,
    xleftmargin=1em,
    framexleftmargin=1em,
    tabsize=4,
    emph={[1]activate,active,f1,f2,f3,range},
    emphstyle={[1]\color{pyfunc}},
    emph={[2]condition1,condition2,condition3},
    emphstyle={[2]\color{pycond}}
}

\definecolor{darkred}{rgb}{0.6,0.0,0.0}
\definecolor{darkgreen}{rgb}{0,0.50,0}
\definecolor{lightblue}{rgb}{0.0,0.42,0.91}
\definecolor{orange}{rgb}{0.99,0.48,0.13}
\definecolor{grass}{rgb}{0.18,0.80,0.18}
\definecolor{pink}{rgb}{0.97,0.15,0.45}

\lstdefinestyle{colored}{ %
  basicstyle=\ttfamily,
  backgroundcolor=\color{white},
  commentstyle=\color{green}\itshape,
  keywordstyle=\color{blue}\bfseries,
  stringstyle=\color{red},
}
\lstdefinelanguage{PythonPlus}[]{Python}{
  morekeywords=[1]{,as,assert,nonlocal,with,yield,self,True,False,None,} 
  morekeywords=[2]{,__init__,__add__,__mul__,__div__,__sub__,__call__,__getitem__,__setitem__,__eq__,__ne__,__nonzero__,__rmul__,__radd__,__repr__,__str__,__get__,__truediv__,__pow__,__name__,__future__,__all__,}, 
  morekeywords=[3]{,object,type,isinstance,copy,deepcopy,zip,enumerate,reversed,list,set,len,dict,tuple,range,xrange,append,execfile,real,imag,reduce,str,repr,}, 
  morekeywords=[4]{,Exception,NameError,IndexError,SyntaxError,TypeError,ValueError,OverflowError,ZeroDivisionError,}, 
  morekeywords=[5]{,ode,fsolve,sqrt,exp,sin,cos,arctan,arctan2,arccos,pi, array,norm,solve,dot,arange,isscalar,max,sum,flatten,shape,reshape,find,any,all,abs,plot,linspace,legend,quad,polyval,polyfit,hstack,concatenate,vstack,column_stack,empty,zeros,ones,rand,vander,grid,pcolor,eig,eigs,eigvals,svd,qr,tan,det,logspace,roll,min,mean,cumsum,cumprod,diff,vectorize,lstsq,cla,eye,xlabel,ylabel,squeeze,}, 
}
\lstdefinelanguage{PyBrIM}[]{PythonPlus}{
  emph={d,E,a,Fc28,Fy,Fu,D,des,supplier,Material,Rectangle,PyElmt},
}
\lstdefinestyle{colorEX}{
  basicstyle=\ttfamily,
  backgroundcolor=\color{white},
  commentstyle=\color{darkgreen}\slshape,
  keywordstyle=\color{blue}\bfseries,
  keywordstyle=[2]\color{blue}\bfseries,
  keywordstyle=[3]\color{grass},
  keywordstyle=[4]\color{red},
  keywordstyle=[5]\color{orange},
  stringstyle=\color{darkred},
  emphstyle=\color{pink}\underbar,
}
\colorlet{literatecolour}{magenta!90!black}

\lstdefinestyle{P}{emph={[2]c1,c2,c3},
    emphstyle={[2]\color{blue}},
    emph={[3]True},
     emph={[1]activate,active,f1,f2,f3,range},
    emphstyle={[1]\color{pyfunc}},
    emphstyle={[3]\color{blue!50!green}},
    literate=*%
{\%}{{\literatecolour:}}{1}%
{:}{{\literatecolour:}}{1}%
{=}{{\literatecolour=}}{1}%
{-}{{\literatecolour-}}{1}%
{+}{{\literatecolour+}}{1}%
{*}{{\literatecolour*}}{1}%
{**}{{\literatecolour{**}}}2%
{/}{{\literatecolour/}}{1}%
{//}{{\literatecolour{//}}}2%
{!}{{\literatecolour!}}{1}%
{[}{{\literatecolour[}}{1}%
{]}{{\literatecolour]}}{1}%
{<}{{\literatecolour<}}{1}%
{>}{{\literatecolour>}}{1}%
{>>>}{\pythonprompt}{3}%
,%
}

\newcommand*{\pythin}{\lstinline[style=P,keepspaces=true]}

\lstnewenvironment{pythonPrime}[1][]{\lstset{style=pythonhighlight-style,backgroundcolor=\color{yellow!5!white}, emph={[2]c1,c2,c3},
    emphstyle={[2]\color{blue}},
    emph={[3]True, switch, case},
     emph={[1]activate,active,f1,f2,f3,range},
    emphstyle={[1]\color{pyfunc}},
    emphstyle={[3]\color{blue!50!green}}
    }}{}

\definecolor{rgreen}{RGB}{10,180,20}

\usepackage{framed}

\crefname{equation}{Eq.}{Eqs.}

\title{Near-Optimal Encodings of Cardinality Constraints} 

\author{Andrew Krapivin\,\orcidlink{0009-0003-0227-7660}}
\author{Benjamin Przybocki\,\orcidlink{0009-0007-5489-1733}}
\author{Bernardo Subercaseaux\,\orcidlink{0000-0003-2295-1299}}

\affil{Carnegie Mellon University \\ {\upshape\ttfamily[akrapivi,bprzyboc,bsuberca]@andrew.cmu.edu}}
\date{}

\newcommand{\eo}{\textsf{ExactlyOne}}

\newcommand{\lrp}[1]{\left( #1 \right)}

\newcommand{\ov}{\mathsf{ov}}

\begin{document}

\maketitle

\begin{abstract}
We present several novel encodings for cardinality constraints, which use fewer clauses than previous encodings 
and, more importantly, introduce new generally applicable techniques for constructing compact encodings.
First, we present a CNF encoding for the $\amo(x_1,\dots,x_n)$ constraint using $2n + 2 \sqrt{2n} + O(\sqrt[3]{n})$ clauses, thus refuting the conjectured optimality of Chen's product encoding. Our construction also yields a smaller monotone circuit for the threshold-2 function, improving on a 50-year-old construction of Adleman and incidentally solving a long-standing open problem in circuit complexity. 
On the other hand, we show that any encoding for this constraint requires at least $2n + \sqrt{n+1} - 2$ clauses, which is the first nontrivial unconditional lower bound for this constraint and answers a question of Ku{\v c}era, Savick{\'{y}}, and Vorel.
We then turn our attention to encodings of $\amk(x_1,\dots,x_n)$, where we introduce
\emph{grid compression}, a technique inspired by hash tables, to give encodings using $2n + o(n)$ clauses as long as $k = o(\sqrt[3]{n})$ and $4n + o(n)$ clauses as long as $k = o(n)$. Previously, the smallest known encodings were of size $(k+1)n + o(n)$ for $k \le 5$ and $7n - o(n)$ for $k \ge 6$.
\end{abstract}

\section{Introduction} \label{sec-intro}

Cardinality constraints are a fundamental building block used to encode problems into conjunctive normal form (CNF). Due to their ubiquitous nature and importance to SAT solving, cardinality constraints have been extensively studied by the SAT community from both theoretical and experimental perspectives (see, e.g., \cite{totalizer,sinz,robust,Chen2010ANS,product-k,cardinality-networks,hash,lower-bound,empirical-study,cnf-symmetric,reeves}).

The most basic cardinality constraint is $\amo(x_1,\dots,x_n)$, which asserts that at most one boolean variable $x_i$ is true. While this constraint can be encoded using $\binom{n}{2}$ clauses:
\[
    \amo(x_1,\dots,x_n) \Longleftrightarrow \bigwedge_{1 \le i < j \le n} (\overline{x_i} \lor \overline{x_j}),
\]
this quadratic blowup in the number of clauses is problematic when $n$ is large. Fortunately, by introducing auxiliary variables, there are several encodings for $\amo(x_1,\dots,x_n)$ using only $O(n)$ clauses, such as the sequential counter encoding from \cite{sinz}. Empirically, SAT solvers perform much better when using linear encodings rather than the quadratic one~\cite{empirical-study}. Which cardinality constraint encoding is the most performant in practice depends on a number of factors, including the specific solver used and the encoding of the remaining constraints, but we focus on three factors that are amenable to theoretical analysis: the number of clauses in the encoding, the number of auxiliary variables, and whether the encoding is propagation complete (also known as arc consistent~\cite{arc-consistent}), a notion defined in~\Cref{sec-prelim}.

Prior to the present work, the smallest known encoding for $\amo(x_1,\dots,x_n)$ was Chen's product encoding~\cite{Chen2010ANS}, which uses $2n + 4 \sqrt{n} + O(\sqrt[4]{n})$ clauses and $2 \sqrt{n} + O(\sqrt[4]{n})$ auxiliary variables while also being propagation complete. Chen conjectured his encoding to be optimal with respect to the number of clauses. Our first contribution is to refute this conjecture:
\begin{restatable}[Multipartite Encoding]{theorem}{thmamo} \label{thm:am1}
    There exists a propagation-complete encoding for the\break $\amo(x_1, \ldots, x_n)$ constraint using $2n + 2 \sqrt{2n} + O(\sqrt[3]{n})$ clauses and $\sqrt{2n} + O(\sqrt[3]{n})$ auxiliary variables.
\end{restatable}
Not only does our encoding improve upon the product encoding with respect to both the number of clauses and the number of auxiliary variables, but it also has a few conceptually interesting aspects. Perhaps surprisingly, our encoding includes an $\amt$ constraint within its definition, which makes it 3-CNF despite the fact that $\amo$ is a 2-CNF function. While all previous linear encodings for this constraint in the literature are 2-CNF, Ku{\v c}era, Savick{\'{y}}, and Vorel~\cite{lower-bound} asked whether optimal propagation-complete encodings of $\amo$ are 2-CNF; our result gives evidence to the contrary. Moreover, as we show in \Cref{sec-circuit}, our construction yields a smaller monotone boolean circuit for the \emph{threshold-2} function, the negation of $\amo$, which improves on a 50-year-old construction of Adleman,\footnote{Adleman's construction was first mentioned in print by Bloniarz~\cite{bloniarz1979}.} answers a 47-year-old open question from Bloniarz~\cite[p.~158]{bloniarz1979}, and has implications for the so-called single-level conjecture~\cite{lenz-wegener,bublitz,mirwald-schnorr,amano-maruoka,single-level}.

Our encoding originates from reinterpreting Chen's product encoding, traditionally presented in terms of a grid, as involving a complete bipartite graph. With this change of perspective, any graph with $n$ edges can form the basis of an encoding of $\amo$, with $2n$ clauses that relate each edge to a pair of auxiliary variables corresponding to its endpoints and additional constraints depending on the graph structure.
Our improvement in the second-order term comes from using a complete multipartite graph with $\omega(1)$ parts, which has edge density $1$ rather than $1/2$ as in a complete bipartite graph; this allows for fewer vertices (i.e., auxiliary variables) for the same number of edges.

With regard to lower bounds, Ku{\v c}era, Savick{\'{y}}, and Vorel~\cite{lower-bound} proved that every propagation-complete encoding of $\amo(x_1, \ldots, x_n)$ requires $2n + \sqrt{n} - 2$ clauses for $n \ge 7$, so \Cref{thm:am1} is close to optimal. On the other hand, prior to the present work, no nontrivial lower bound was known without assuming propagation completeness.\footnote{Sinz~\cite{sinz} claimed that every encoding of $\amk(x_1,\dots,x_n)$ has at least $n$ clauses for all $k \in [n-2]$, but his proof has an irreparable gap. Indeed, the claim is false in general: Ben-Haim, Ivrii, Margalit, and Matsliah~\cite{hash} showed that $\amk(x_1,\dots,x_n)$ can be encoded using $O(\log n)$ clauses when $k = n-O(1)$.} Our second contribution provides such a bound---which is in fact better than the bound assuming propagation completeness---and answers a question from \cite{lower-bound}:
\begin{restatable}{theorem}{thmlowerbound} \label{thm:lower-bound}
    Every encoding of the $\amo(x_1, \ldots, x_n)$ constraint has at least $2n + \sqrt{n+1} - 2$ clauses for $n \ge 8$.
\end{restatable}
Together with \Cref{thm:am1} (or Chen's result), this implies that the minimum number of clauses in an encoding of $\amo(x_1, \ldots, x_n)$ is $2n + \Theta(\sqrt{n})$.

Next, we turn to the constraint $\amk(x_1,\dots,x_n)$, which asserts that at most $k$ of the boolean variables $x_1,\dots,x_n$ are true. Sinz~\cite{sinz} gave an encoding for this using $7n - 3 \floor{\log n} - 6$ clauses and $2n-2$ auxiliary variables. When $k$ is large relative to $n$, this is the smallest known encoding for $\amk(x_1,\dots,x_n)$ as measured by the number of clauses. Our third contribution is to present smaller encodings when $k$ is small relative to $n$:
\begin{restatable}[Grid Compression]{theorem}{thmgridcompression} \label{thm:conjunctive-compression}
    There is an encoding of the $\amk(x_1,\dots,x_n)$ constraint using $4n + O\lrp{\sqrt[3]{kn^2}}$ clauses and $O\lrp{\sqrt[3]{kn^2}}$ auxiliary variables.
\end{restatable}
\begin{restatable}[Disjunctive Grid Compression]{theorem}{thmdisjunctivegridcompression} \label{thm:disjunctive-compression}
    There is an encoding of the $\amk(x_1,\dots,x_n)$ constraint using $2n + O\lrp{\sqrt{nk^3\log^3_k n}}$ clauses and $O\lrp{\sqrt{nk^3\log^3_k n}}$ auxiliary variables.
\end{restatable}
In particular, for $k=o(n)$, \Cref{thm:conjunctive-compression} gives an encoding using $4n + o(n)$ clauses and $o(n)$ auxiliary variables; for $k = o(\sqrt[3]{n})$, \Cref{thm:disjunctive-compression} gives an encoding using $2n + o(n)$ clauses and $o(n)$ auxiliary variables. We show in \Cref{cor:lower-bound-amk} that encoding $\amk(x_1,\dots,x_n)$ requires $2(n-k) + \Omega(\sqrt{n-k})$ clauses, so \Cref{thm:disjunctive-compression} notably implies that the minimum number of clauses in an encoding of $\amk(x_1,\dots,x_n)$ is $2n + \widetilde{\Theta}(\sqrt{n})$ for any $k = \log^{O(1)} n$.

Unfortunately, neither our encodings for $\amk$ nor Sinz's encoding are propagation complete. The smallest known propagation-complete encoding when $k$ is a small constant is the \emph{generalized product encoding}~\cite{product-k}, which uses $(k+1)n + O(k^2 n^{k/(k+1)})$ clauses and $O(kn^{k/(k+1)})$ auxiliary variables for $k=o(\log n / \log \log n)$. For large $k$, the smallest known propagation-complete encoding combines the AKS sorting network~\cite{aks} and a CNF encoding of sorting networks~\cite{sorting-networks,cardinality-networks}, and it uses $O(n \log k)$ clauses and auxiliary variables. Nonetheless, the disjunctive grid compression encoding turns out to be competitive in practice for some benchmarks. While our motivation for this work was primarily theoretical, these preliminary empirical results highlight that our techniques can be practical, and they cast doubt on the conventional wisdom that propagation completeness is essential for performance.  We discuss this further in~\Cref{sec:discussion} and Appendix~\ref{sec:empirical}.

The encodings from~\Cref{thm:conjunctive-compression} and~\Cref{thm:disjunctive-compression} are based on a new technique that we call \emph{grid compression}. The technique consists of arranging the input variables into a grid $M$, which is then ``compressed'' into a smaller grid $L$. Then, we enforce that at most $k$ variables in $L$ are true, which is cheaper since $L$ is smaller. The techniques to map $M$ into $L$ are inspired by those used in the construction of hash tables~\cite[Chapter~11]{clrs}.

On top of grid compression, \Cref{thm:disjunctive-compression} leverages a novel paradigm we call \emph{disjunctive switching}, which aims to address a common problem that arises when translating algorithmic ideas into CNF encodings. When an algorithm can take different branches according to some condition (i.e., a \emph{switch}), its runtime will only depend on the taken branch, whereas na\"ive CNF encodings include clauses for each possible branch, and thus their size is based on the sum of all branches. In a nutshell, disjunctive switching introduces wide clauses stating that \emph{some} branch will be taken, and then succinctly enforces that branches whose condition is not met cannot be taken, thus making the taken branch the only way to satisfy the disjunction. Besides applying disjunctive switching to the disjunctive grid compression encoding, we illustrate its utility by showing how the generalized product encoding for $\amk(x_1,\dots,x_n)$~\cite{product-k}, which uses $(k+1)n + O(k^2 n^{k/(k+1)})$ clauses for $k=o(\log n / \log \log n)$, can be ``disjunctivized'' into an encoding with only $2n + O(kn^{k/(k+1)})$ clauses (see \Cref{thm:amk-disjunctive-product}). If not for \Cref{thm:disjunctive-compression}, this would be the smallest known encoding of this constraint for $k=o(\log n / \log \log n)$.

\subparagraph*{Proofs.} Several of our results require technical combinatorial proofs, which are postponed to the appendix for improved readability.


\section{Preliminaries} \label{sec-prelim}
%


A \emph{CNF formula} is a set of clauses, each of which is a set of literals, and the \emph{size} of such a formula is the number of clauses it contains. Given a set of variables $X$, we denote by $\mathrm{lit}(X) := X \cup \{ \overline{x} \mid x \in X\}$ the set of literals over those variables.
For a partial assignment $\tau: \mathcal{V} \to \{\bot, \top\}$ and a formula $\varphi$, we denote by $\varphi|_\tau$ the formula obtained by eliminating from $\varphi$ each clause satisfied by $\tau$, and then from each remaining clause eliminating every literal $\ell$ such that $\tau \models \overline{\ell}$. We will write $\textsf{SAT}(\varphi)$ to say that $\tau \models \varphi$ for some assignment $\tau$.

We now present the two main definitions used in this paper.

\begin{definition}[CNF Encoding]
    Given a boolean function $f \colon \{\bot,\top\}^n \to \{\bot, \top\}$, and a sequence of propositional variables $X := (x_1, \dots, x_n)$, we say that a CNF formula $\varphi$ over variables $X \sqcup Y$ \emph{encodes} $f$ if, for every assignment $\tau$ of $X$,
    \[
    f(\tau(x_1), \dots, \tau(x_n)) = \top \iff  \mathsf{SAT}(\varphi|_\tau).
    \]
    The variables in $Y$ are called \emph{auxiliary} variables, whereas those in $X$ are called \emph{input} variables.
\end{definition}
%
%
\begin{definition}[Propagation Completeness {\cite{lower-bound}}]
    Let $f \colon \{\bot,\top\}^n \to \{\bot,\top\}$ be a boolean function and $\varphi$ a CNF encoding for $f$ over input variables $X$ and auxiliary variables $Y$. We say that $\varphi$ is \emph{propagation complete} if for every $\ell_1,\dots,\ell_m \in \mathrm{lit}(X)$
    \[
       \left(\varphi \land \bigwedge_{i = 1}^{m-1} \ell_i\right)  \models \ell_m \; \; \implies  \left[\left(\varphi \land \bigwedge_{i=1}^{m-1} \ell_i\right) \vdash_1 \ell_m  \; \; \text{or} \; \;  \left(\varphi \land \bigwedge_{i=1}^{m-1} \ell_i\right) \vdash_1 \bot\right],
    \]
    where $\psi \vdash_1 \ell$ means that $\ell$ can be derived from $\psi$ by unit propagation.
\end{definition}

For example, to say that an encoding $\varphi$ of $\amo(x_1,\dots,x_n)$ is propagation complete means that, whenever a variable $x_i$ is assigned to true, all other variables $x_j$ will be assigned to false by unit propagation, and thus if two input variables are ever assigned to true, a contradiction is detected by unit propagation alone. Given that SAT solvers perform unit propagation very efficiently, propagation-complete encodings guarantee that entailments involving the encoded function can be derived efficiently by a solver.

\section{\texorpdfstring{Encodings for $\amo$}{Encodings for AMO}} \label{sec-amo}

In this section, we prove \Cref{thm:am1,thm:lower-bound}. We start by presenting Chen's product encoding for $\amo$ and giving a new graph-theoretic perspective on it, from which our improved encoding will seem more natural.

\subsection{Two perspectives on the product encoding}\label{sec-perspective}

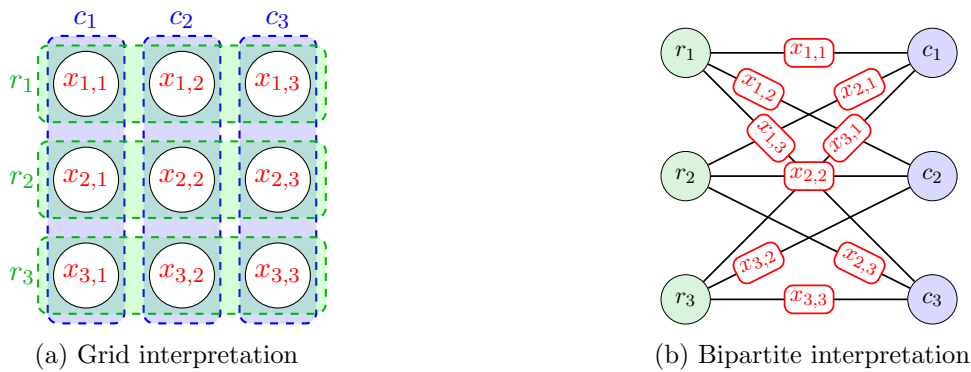
\begin{figure}
    \begin{subfigure}{0.48\linewidth}
        \centering
    \begin{tikzpicture}[scale=0.85]
  \draw[blue,dashed,rounded corners,thick,fill=blue,fill opacity=0.15] (0.15000000000000002,0) rectangle (1.35,4.5);
  \node[text=blue] at (0.75,4.75) {$c_{1}$};
  \draw[blue,dashed,rounded corners,thick,fill=blue,fill opacity=0.15] (1.6500000000000001,0) rectangle (2.8499999999999996,4.5);
  \node[text=blue] at (2.25,4.75) {$c_{2}$};
  \draw[blue,dashed,rounded corners,thick,fill=blue,fill opacity=0.15] (3.1500000000000004,0) rectangle (4.35,4.5);
  \node[text=blue] at (3.75,4.75) {$c_{3}$};
  \definecolor{rgreen}{RGB}{10,180,20};
  \draw[rgreen,dashed,rounded corners,thick,fill=green,fill opacity=0.15] (0,0.15000000000000002) rectangle (4.5,1.35);
  \node[text=rgreen] at (-0.25,0.75) {$r_{3}$};
  \draw[rgreen,dashed,rounded corners,thick,fill=green,fill opacity=0.15] (0,1.6500000000000001) rectangle (4.5,2.8499999999999996);
  \node[text=rgreen] at (-0.25,2.25) {$r_{2}$};
  \draw[rgreen,dashed,rounded corners,thick,fill=green,fill opacity=0.15] (0,3.1500000000000004) rectangle (4.5,4.35);
  \node[text=rgreen] at (-0.25,3.75) {$r_{1}$};
  \node[text=red,draw,circle,minimum size=6mm,fill=white, inner sep=2pt] at (0.75,3.75) {$x_{1,1}$};
  \node[text=red,draw,circle,minimum size=6mm,fill=white, inner sep=2pt] at (0.75,2.25) {$x_{2,1}$};
  \node[text=red,draw,circle,minimum size=6mm,fill=white, inner sep=2pt] at (0.75,0.75) {$x_{3,1}$};
  \node[text=red,draw,circle,minimum size=6mm,fill=white, inner sep=2pt] at (2.25,3.75) {$x_{1,2}$};
  \node[text=red,draw,circle,minimum size=6mm,fill=white, inner sep=2pt] at (2.25,2.25) {$x_{2,2}$};
  \node[text=red,draw,circle,minimum size=6mm,fill=white, inner sep=2pt] at (2.25,0.75) {$x_{3,2}$};
  \node[text=red,draw,circle,minimum size=6mm,fill=white, inner sep=2pt] at (3.75,3.75) {$x_{1,3}$};
  \node[text=red,draw,circle,minimum size=6mm,fill=white, inner sep=2pt] at (3.75,2.25) {$x_{2,3}$};
  \node[text=red,draw,circle,minimum size=6mm,fill=white, inner sep=2pt] at (3.75,0.75) {$x_{3,3}$};
\end{tikzpicture}
    \caption{Grid interpretation}
    \label{fig:grid}
    \end{subfigure}
    \hfill
      \begin{subfigure}{0.48\linewidth}
        \centering
    \scalebox{0.82}{
\begin{tikzpicture}[]
  \node[draw,circle,minimum size=8mm,fill=rgreen, fill opacity=0.15, text opacity=1] at (-2,-2.0) (node1) {$r_{3}$};
  \node[draw,circle,minimum size=8mm,fill=rgreen, fill opacity=0.15, text opacity=1] at (-2,0.0) (node2) {$r_{2}$};
  \node[draw,circle,minimum size=8mm,fill=rgreen, fill opacity=0.15, text opacity=1] at (-2,2.0) (node3) {$r_{1}$};
  \node[draw,circle,minimum size=8mm,fill=blue, fill opacity=0.15, text opacity=1] at (2,-2.0) (node4) {$c_{3}$};
  \node[draw,circle,minimum size=8mm,fill=blue, fill opacity=0.15, text opacity=1] at (2,0.0) (node5) {$c_{2}$};
  \node[draw,circle,minimum size=8mm,fill=blue, fill opacity=0.15, text opacity=1] at (2,2.0) (node6) {$c_{1}$};
  \draw[thick,color=black] (node3) -- (node6) node [color=red,midway,fill=white,rectangle,draw,rounded corners,sloped, inner sep=2.5pt] {$x_{1, 1}$};
  \draw[thick,color=black] (node3) -- (node5) node [color=red,near start,fill=white,rectangle,draw,rounded corners,sloped, inner sep=2.5pt] {$x_{1, 2}$};

  \draw[thick,color=black] (node2) -- (node6) node [color=red, near end,fill=white,rectangle,draw,rounded corners,sloped, inner sep=2.5pt] {$x_{2, 1}$};

    \draw[thick,color=black] (node3) -- (node4) node [color=red, pos=0.32,fill=white,rectangle,draw,rounded corners,sloped, inner sep=2.5pt] {$x_{1, 3}$};

  \draw[thick,color=black] (node2) -- (node4) node [color=red, near end,fill=white,rectangle,draw,rounded corners,sloped, inner sep=2.5pt] {$x_{2, 3}$};
  \draw[thick,color=black] (node1) -- (node6) node [color=red,pos=0.68,fill=white,rectangle,draw,rounded corners,sloped, inner sep=2.5pt] {$x_{3, 1}$};
  
  \draw[thick,color=black] (node1) -- (node5) node [color=red,near start,fill=white,rectangle,draw,rounded corners,sloped, inner sep=2.5pt] {$x_{3, 2}$};
  \draw[thick,color=black] (node1) -- (node4) node [color=red,midway,fill=white,rectangle,draw,rounded corners,sloped, inner sep=2.5pt] {$x_{3, 3}$};

    \draw[thick,color=black] (node2) -- (node5) node [color=red,midway,fill=white,rectangle,draw,rounded corners,sloped, inner sep=2.5pt] {$x_{2, 2}$};
\end{tikzpicture}
}
    \caption{Bipartite interpretation}
    \label{fig:bipartite}
    \end{subfigure}
    \caption{Illustration of our proposed change of perspective on Chen's product encoding}
\end{figure}

The traditional way to present Chen's product encoding is by arranging the input variables into a grid as depicted in \Cref{fig:grid}. Then, the key insight underlying the encoding is that at most one input variable is true if and only if (a) at most one row contains a true variable and (b) at most one column contains a true variable. This is encoded as follows. For each row, introduce an auxiliary variable $r_i$ that is implied by each variable in that row, and do similarly for the columns with auxiliary variables $c_j$. Then, one recursively encodes the $\amo$ constraints for $\{r_1,r_2,\dots\}$ and $\{c_1,c_2,\dots\}$.

More formally, rename the input variables $x_1,\dots,x_n$ to be of the form $x_{i,j}$ with $i,j \in [p]$, where $p = \ceil{\sqrt{n}}$. Let $E$ be the set of ordered pairs $(i,j)$ to which some variable is assigned. Then, the product encoding is as follows:
\[
    \textsf{PE}(\{ x_{i, j} \mid (i, j) \in E\}) := \left(\bigwedge_{(i, j) \in E} (\overline{x_{i, j}} \lor r_i) \land (\overline{x_{i, j}} \lor c_j) \right) \land \textsf{PE}(r_1, \ldots, r_p) \land \textsf{PE}(c_1, \ldots, c_p).
\]
For the base case (say, when $n \le 4$), we use the direct encoding: $\bigwedge_{1 \le i < j \le n} (\overline{x_i} \lor \overline{x_j})$.

But there is also another interpretation of the product encoding, illustrated in \Cref{fig:bipartite}. Here, we identify the input variables with the edges of a bipartite graph. We say that an edge is \emph{selected} if the corresponding input variable is true, and we say that a vertex is \emph{selected} if it is incident to a selected edge. Now, the key insight can be rephrased as follows: at most one edge is selected if and only if (a) at most one vertex in the left part is selected and (b) at most one vertex in the right part is selected. Of course, the grid interpretation and bipartite interpretation are equivalent, but the bipartite interpretation provides a nice conceptual lens for designing new encodings for $\amo$.

Given a graph $G$, our goal is to encode $\amo(E(G))$, the constraint asserting that at most one edge is selected. Let the input variables be $x_{\{i,j\}}$ for each $\{i,j\} \in E(G)$. As in the product encoding, we introduce an auxiliary variable $v_i$ for each vertex $i \in V(G)$, and we spend $2|E|$ clauses to make each edge imply its endpoints: $\overline{x_{i,j}} \lor v_i$ and $\overline{x_{i,j}} \lor v_j$. Then, it remains to use the auxiliary variables associated with the vertices to encode that at most one edge is selected; how we do this depends on the choice of $G$. The efficiency of the encoding (i.e., how many clauses it has as a function of $|E|$) depends on the edge density of the graph and how succinctly we can use the vertex variables to encode that at most one edge is selected.

\subsection{The multipartite encoding} \label{sec-multipartite}

We now describe an encoding for $\amo(x_1,\dots,x_n)$ with $2n+2\sqrt{2n}+O(\sqrt[3]{n})$ clauses and $\sqrt{2n}+O(\sqrt[3]{n})$ auxiliary variables, thus proving~\Cref{thm:am1}. We use the graph-theoretic strategy just described, taking $G$ to be a complete multipartite graph with $\Theta(\sqrt[6]{n})$ parts and $\Theta(\sqrt[3]{n})$ vertices within each part; we therefore call our encoding the \emph{multipartite encoding}. The constants are chosen so that $G$ has ${\sim}\sqrt{2n}$ vertices and ${\sim}n$ edges. These parameters turn out to be a good choice for two reasons. First, $G$ has a high edge density, which allows us to assign more input variables to the edges of $G$. Second, we have a succinct way to use the vertex variables to encode that at most one edge is selected:
\begin{framed}
\noindent
\textbf{Key insight:} At most one edge is selected if and only if (a) at most one vertex from each part is selected and (b) at most two parts contain a selected vertex (see \Cref{fig:multipartite}).
\end{framed}

\begin{figure}
    \begin{subfigure}{0.48\linewidth}
        \centering
    \scalebox{0.9}{
\begin{tikzpicture}[
    scale=0.6,
    vertex/.style={circle, draw=black, fill=white, inner sep=0pt, minimum size=7pt},
    edge/.style={line width=0.45pt, draw=black!60, opacity=0.5},
    blob/.style={rounded corners=10pt, very thick, inner sep=10pt, draw opacity=0.9, fill opacity=0.14}
]
\coordinate (tl) at (-3,  3);
\coordinate (tr) at ( 3,  3);
\coordinate (bl) at (-3, -3);
\coordinate (br) at ( 3, -3);

\foreach \name/\base in {a/tl,b/tr,c/bl,d/br}{
  \foreach \i/\x/\y in {1/-0.6/0.6, 2/0.6/0.6, 3/-0.6/-0.6, 4/0.6/-0.6}{
    \node[vertex] (\name\i) at ($(\base)+(\x,\y)$) {};
  }
}

\def\A{a1,a2,a3,a4}
\def\AONE{a1,a2}
\def\ATWO{a3,a4}
\def\ACONE{a1,a3}
\def\ACTWO{a2,a4}
\def\B{b1,b2,b3,b4}
\def\BONE{b1, b2}
\def\BTWO{b3, b4}
\def\BCONE{b1,b3}
\def\BCTWO{b2,b4}
\def\C{c1,c2,c3,c4}
\def\CONE{c1,c2}
\def\CTWO{c3, c4}
\def\CCONE{c1, c3}
\def\CCTWO{c2, c4}
\def\D{d1,d2,d3,d4}
\def\DONE{d1,d2}
\def\DTWO{d3,d4}
\def\DCONE{d1, d3}
\def\DCTWO{d2, d4}

\def\AL{a2,a3,a4}
\def\BL{b1,b3,b4}
\def\CL{c1,c2,c4}
\def\DL{d1,d2,d3}

\begin{scope}[on background layer]
  \node[blob, draw=teal!60,   fill=teal!50,   fit=(a1)(a2)(a3)(a4)] {};
  \node[blob, draw=teal!60,    fill=teal!50,    fit=(b1)(b2)(b3)(b4)] {};
  \node[blob, draw=teal!60,   fill=teal!50,   fit=(c1)(c2)(c3)(c4)] {};
  \node[blob, draw=teal!60, fill=teal!50, fit=(d1)(d2)(d3)(d4)] {};

    \foreach \u in \AL { \foreach \v in \DL { \draw[edge] (\u) -- (\v); } }

     \foreach \u in \BL { \foreach \v in \CL { \draw[edge] (\u) -- (\v); } }
    

  \foreach \u in \AONE { \foreach \v in \BONE { \draw[edge] (\u) to[bend left=10] (\v); } }
  \foreach \u in \AONE { \foreach \v in \BTWO { \draw[edge] (\u) -- (\v); } }
  \foreach \u in \ATWO { \foreach \v in \BONE { \draw[edge] (\u) -- (\v); } }
  
  \foreach \u in \ATWO { \foreach \v in \BTWO { \draw[edge] (\u) to[bend right=10] (\v); } }

  \foreach \u in \ACONE { \foreach \v in \CCONE { \draw[edge] (\u) to[bend right=10] (\v); } }
  \foreach \u in \ACONE { \foreach \v in \CCTWO { \draw[edge] (\u) -- (\v); } }
  \foreach \u in \ACTWO { \foreach \v in \CCONE { \draw[edge] (\u) -- (\v); } }
  \foreach \u in \ACTWO { \foreach \v in \CCTWO { \draw[edge] (\u) to[bend left=10] (\v); } }

  \foreach \u in \CONE { \foreach \v in \DONE { \draw[edge] (\u) to[bend left=10] (\v); } }
  \foreach \u in \CONE { \foreach \v in \DTWO { \draw[edge] (\u) -- (\v); } }
  \foreach \u in \CTWO { \foreach \v in \DONE { \draw[edge] (\u) -- (\v); } }
  \foreach \u in \CTWO { \foreach \v in \DTWO { \draw[edge] (\u) to[bend right=10] (\v); } }

  \foreach \u in \BCONE { \foreach \v in \DCONE { \draw[edge] (\u) to[bend right=10] (\v); } }
  \foreach \u in \BCONE { \foreach \v in \DCTWO { \draw[edge] (\u) -- (\v); } }
  \foreach \u in \BCTWO { \foreach \v in \DCONE { \draw[edge] (\u) -- (\v); } }
  \foreach \u in \BCTWO { \foreach \v in \DCTWO { \draw[edge] (\u) to[bend left=10] (\v); } }
  

  \draw[edge] (a1) to[bend left=10] (d4);
  \draw[edge] (c3) to[bend left=10] (b2);
\end{scope}

\draw[edge, ultra thick, magenta,opacity=0.8] (a2) -- (c1);
\node[vertex, fill=black] (a22) at (a2) {};
\node[vertex, fill=black] (c11) at (c1) {};

\draw[edge, ultra thick, magenta, opacity=0.8] (b3) -- (d4);
\node[vertex, fill=black] (b33) at (b3) {};
\node[vertex, fill=black] (d44) at (d4) {};



\node[] at (tl) {\Large $P_1$};
\node[] at (tr) {\Large $P_2$};
\node[] at (bl) {\Large $P_3$};
\node[] at (br) {\Large $P_4$};

\end{tikzpicture}
    }
    \caption{$\textsf{AtMostTwo}$-parts constraint is violated.}
    \end{subfigure}
      \begin{subfigure}{0.48\linewidth}
        \centering
    \scalebox{0.9}{\begin{tikzpicture}[
    scale=0.6,
    vertex/.style={circle, draw=black, fill=white, inner sep=0pt, minimum size=7pt},
    edge/.style={line width=0.45pt, draw=black!60, opacity=0.5},
    blob/.style={rounded corners=10pt, very thick, inner sep=10pt, draw opacity=0.9, fill opacity=0.14}
]
\coordinate (tl) at (-3,  3);
\coordinate (tr) at ( 3,  3);
\coordinate (bl) at (-3, -3);
\coordinate (br) at ( 3, -3);

\foreach \name/\base in {a/tl,b/tr,c/bl,d/br}{
  \foreach \i/\x/\y in {1/-0.6/0.6, 2/0.6/0.6, 3/-0.6/-0.6, 4/0.6/-0.6}{
    \node[vertex] (\name\i) at ($(\base)+(\x,\y)$) {};
  }
}

\def\A{a1,a2,a3,a4}
\def\B{b1,b2,b3,b4}
\def\C{c1,c2,c3,c4}
\def\D{d1,d2,d3,d4}

\def\A{a1,a2,a3,a4}
\def\AONE{a1,a2}
\def\ATWO{a3,a4}
\def\ACONE{a1,a3}
\def\ACTWO{a2,a4}
\def\B{b1,b2,b3,b4}
\def\BONE{b1, b2}
\def\BTWO{b3, b4}
\def\BCONE{b1,b3}
\def\BCTWO{b2,b4}
\def\C{c1,c2,c3,c4}
\def\CONE{c1,c2}
\def\CTWO{c3, c4}
\def\CCONE{c1, c3}
\def\CCTWO{c2, c4}
\def\D{d1,d2,d3,d4}
\def\DONE{d1,d2}
\def\DTWO{d3,d4}
\def\DCONE{d1, d3}
\def\DCTWO{d2, d4}

\def\AL{a2,a3,a4}
\def\BL{b1,b3,b4}
\def\CL{c1,c2,c4}
\def\DL{d1,d2,d3}

\begin{scope}[on background layer]
  \node[blob, draw=gray!60,   fill=gray!50,   fit=(a1)(a2)(a3)(a4)] {};
  \node[blob, draw=red!60,    fill=red!50,    fit=(b1)(b2)(b3)(b4)] {};
  \node[blob, draw=gray!60,   fill=gray!50,   fit=(c1)(c2)(c3)(c4)] {};
  \node[blob, draw=teal!60, fill=teal!50, fit=(d1)(d2)(d3)(d4)] {};

    \foreach \u in \AL { \foreach \v in \DL { \draw[edge] (\u) -- (\v); } }

     \foreach \u in \BL { \foreach \v in \CL { \draw[edge] (\u) -- (\v); } }

  \foreach \u in \AONE { \foreach \v in \BONE { \draw[edge] (\u) to[bend left=10] (\v); } }
  \foreach \u in \AONE { \foreach \v in \BTWO { \draw[edge] (\u) -- (\v); } }
  \foreach \u in \ATWO { \foreach \v in \BONE { \draw[edge] (\u) -- (\v); } }
  
  \foreach \u in \ATWO { \foreach \v in \BTWO { \draw[edge] (\u) to[bend right=10] (\v); } }

  \foreach \u in \ACONE { \foreach \v in \CCONE { \draw[edge] (\u) to[bend right=10] (\v); } }
  \foreach \u in \ACONE { \foreach \v in \CCTWO { \draw[edge] (\u) -- (\v); } }
  \foreach \u in \ACTWO { \foreach \v in \CCONE { \draw[edge] (\u) -- (\v); } }
  \foreach \u in \ACTWO { \foreach \v in \CCTWO { \draw[edge] (\u) to[bend left=10] (\v); } }

  \foreach \u in \CONE { \foreach \v in \DONE { \draw[edge] (\u) to[bend left=10] (\v); } }
  \foreach \u in \CONE { \foreach \v in \DTWO { \draw[edge] (\u) -- (\v); } }
  \foreach \u in \CTWO { \foreach \v in \DONE { \draw[edge] (\u) -- (\v); } }
  \foreach \u in \CTWO { \foreach \v in \DTWO { \draw[edge] (\u) to[bend right=10] (\v); } }

  \foreach \u in \BCONE { \foreach \v in \DCONE { \draw[edge] (\u) to[bend right=10] (\v); } }
  \foreach \u in \BCONE { \foreach \v in \DCTWO { \draw[edge] (\u) -- (\v); } }
  \foreach \u in \BCTWO { \foreach \v in \DCONE { \draw[edge] (\u) -- (\v); } }
  \foreach \u in \BCTWO { \foreach \v in \DCTWO { \draw[edge] (\u) to[bend left=10] (\v); } }

    \draw[edge] (a1) to[bend left=10] (d4);
  \draw[edge] (c3) to[bend left=10] (b2);

\end{scope}



\draw[edge, ultra thick, magenta, opacity=0.8] (d1) to[bend left=10] (b3);
\draw[edge, ultra thick, magenta, opacity=0.8] (d1) -- (b4);

\node[vertex, fill=black] (d11) at (d1) {};
\node[vertex, fill=black] (b33) at (b3) {};
\node[vertex, fill=black] (b44) at (b4) {};

\node[] at (tl) {\Large $P_1$};
\node[] at (tr) {\Large $P_2$};
\node[] at (bl) {\Large $P_3$};
\node[] at (br) {\Large $P_4$};

\end{tikzpicture}
    }
    \caption{$\amo'(\_;z_2)$-vertex constraint is violated.}
    \end{subfigure}
    \caption{Illustration of the multipartite encoding. Parts violating the $\amo$-vertex constraint are shaded red. Parts for which $z_k$ is true are shaded teal.}
    \label{fig:multipartite}
\end{figure}
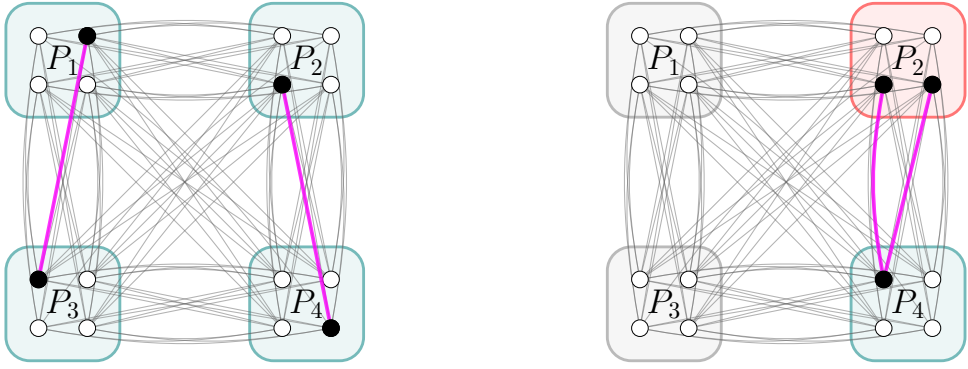

To materialize this encoding, we first require an intermediate construction, which is a simple modification of Chen's product encoding. Let $\amo'(x_1,\dots,x_n;z)$ be the constraint asserting that at most one of the variables $x_1,\dots,x_n$ is true and that $x_i$ implies $z$ for each $i \in [n]$; that is \(\amo'(x_1,\dots,x_n;z) \Longleftrightarrow \amo(x_1, \dots, x_n) \land \bigwedge_{i \in [n]} (\overline{x_i} \lor z) \).

\begin{restatable}{lemma}{lemamoindicator} \label{lem-amo-indicator} 
    There is a propagation-complete encoding of the $\amo'(x_1,\dots,x_n;z)$ constraint using $2n + O(\sqrt{n})$ clauses and $O(\sqrt{n})$ auxiliary variables.
\end{restatable}
\begin{proof}[Proof sketch]
    It suffices to use Chen's product encoding and add one clause $(\overline{r_i} \lor z)$ per row $r_i$. Then, each input variable implies $z$ indirectly through its row.
\end{proof}

\noindent
Now we can prove \Cref{thm:am1}.

\begin{proof}[Proof of \Cref{thm:am1}]
    Let $G$ be the complete $p$-partite graph with $q$ vertices within each part for $p := \ceil{\sqrt[6]{n}}+1$ and $q := \ceil{\sqrt{2} \cdot \sqrt[3]{n}}$. This way, $|E(G)| = \binom{p}{2}q^2 \ge n$. Let $P_1,\dots,P_p$ be the parts of $G$. Assign the variables $x_1, \dots, x_n$ to distinct edges of $G$, renaming the variables so that $x_{\{i,j\}}$ is the variable assigned to the edge $\{i,j\}$. Let $E$ be the set of edges of $G$ to which some variable is assigned. Discard any vertices of $G$ not incident to an edge from $E$. Introduce auxiliary variables $v_i$ for each $i \in V(G)$ and $z_k$ for each $k \in [p]$. Our encoding is as follows:
    \begin{multline*}
        \textsf{ME}(\{x_{\{i,j\}} \mid \{i,j\} \in E\}) := \left(\bigwedge_{\{i,j\} \in E} (\overline{x_{\{i, j\}}} \lor v_i) \land (\overline{x_{\{i, j\}}} \lor v_j) \right) \land \\ \left(\bigwedge_{k \in [p]} \amo'(\{v_i \mid i \in P_k\};z_k) \right)  \land\, \amt(z_1,\dots,z_p),
    \end{multline*}
    where we use the generalized product encoding for $\amt(z_1,\dots,z_p)$---described as well in~\Cref{sec:gdpe}---which is propagation complete and uses $3p + O(p^{2/3})$ clauses and $O(p^{2/3})$ auxiliary variables~\cite{product-k}.

    First, we argue that the encoding is correct and propagation complete. Suppose that at most one input variable is true. If zero input variables are true, then $\textsf{ME}(\{x_{\{i,j\}} \mid \{i,j\} \in E\})$ is satisfiable by setting all of the $v$ and $z$ auxiliary variables to false. Otherwise, exactly one input variable $x_{\{i,j\}}$ is true. Let $k$ and $\ell$ be such that $i \in P_k$ and $j \in P_\ell$. Then, $\textsf{ME}(\{x_{\{i,j\}} \mid \{i,j\} \in E\})$ is satisfiable by setting $v_i$, $v_j$, $z_k$, and $z_\ell$ to true and all of the other $v$ and $z$ auxiliary variables to false.
    
    It remains to show that $\textsf{ME}(\{x_{\{i,j\}} \mid \{i,j\} \in E\}) \land x_{\{i,j\}} \vdash_1 \overline{x_{\{i',j'\}}}$ for all $\{i',j'\} \in E \setminus \{\{i,j\}\}$. If $x_{\{i,j\}}$ is true, then $v_i$ and $v_j$ can be derived by one step of unit propagation, and since our encodings of $\amo'(\{v_i \mid i \in P_k\};z_k)$ and $\amo'(\{v_i \mid i \in P_\ell\};z_\ell)$ are propagation complete by \Cref{lem-amo-indicator}, we can derive $z_k$ and $z_\ell$ by unit propagation (for $k$ and  $\ell$ defined as above), as well as $\overline{v_{i'}}$ for all $i' \in P_k \cup P_\ell \setminus \{v_i,v_j\}$. Thus, we can derive $\overline{x_{\{i',j'\}}}$ by unit propagation for every edge $\{i',j'\}$ between $P_k$ and $P_\ell$ other than $\{i,j\}$. Since our encoding of $\amt(z_1,\dots,z_p)$ is propagation complete, $\overline{z_{k'}}$ is derivable by unit propagation for all $k' \notin \{k,\ell\}$, and therefore $\overline{v_{i'}}$ is derivable by unit propagation for all $i' \in P_{k'}$ and all $k' \notin \{k,\ell\}$. Thus, we can derive $\overline{x_{\{i',j'\}}}$ by unit propagation for every edge $\{i',j'\}$ that is not between $P_k$ and $P_\ell$. We conclude that the encoding is correct and propagation complete.

   Finally, the number of clauses is $2n + p \cdot (2q + O(\sqrt{q})) + (3p + O(p^{2/3})) = 2n + 2 \sqrt{2n} + O(\sqrt[3]{n})$, and the number of auxiliary variables is $pq + p \cdot O(\sqrt{q}) + O(p^{2/3}) = \sqrt{2n} + O(\sqrt[3]{n})$.
\end{proof}

\subsection{The clique encoding}

The analysis in the proof of~\Cref{thm:am1} generalizes as follows: let $f_1(n)$ and $f_2(n)$ be the minimum size of encodings for $\amo$ and $\amt$ on $n$ variables, respectively. Then, taking $G$ to be a complete $p$-partite graph with $q$ vertices within each part yields  
\[
f_1(n) \leq 2n + p \cdot f_1(q) + f_2(p).
\]
For $p, q = \omega(1)$, the best choice is that of~\Cref{thm:am1}.
When $p = 2$, we have $f_2(p) = 0$, and since $G$ is a complete bipartite graph, we recover the product encoding. The other extreme is to take $q = 1$, in which case $f_1(q) = 0$, and $G$ is a complete graph; then, the size of the encoding is $2n + f_2(p)$. Using our disjunctive grid compression encoding (\Cref{thm:disjunctive-compression}) for $\amt$, we can get an encoding for $\amo(x_1,\dots,x_n)$ with slightly fewer clauses than in~\Cref{thm:am1}, at the cost of losing propagation completeness. This is the smallest encoding for $\amo$ that we know of.

\begin{restatable}[Clique Encoding]{theorem}{thmamoclique} \label{thm:am1-clique}
    There is an encoding of the $\amo(x_1, \ldots, x_n)$ constraint using $2n + 2 \sqrt{2n} + \widetilde{O}(\sqrt[4]{n})$ clauses and $\sqrt{2n} + \widetilde{O}(\sqrt[4]{n})$ auxiliary variables.
\end{restatable}

\subsection{A smaller circuit for threshold-2} \label{sec-circuit}

The search for \emph{small} circuits (i.e., using few gates or wires) for cardinality constraints has been a cornerstone of circuit complexity (see, e.g., \cite{bloniarz1979,dunne1984,wegener1987,grinchuk,sergeev}). While the product encoding was only proposed in the context of SAT in 2010~\cite{Chen2010ANS}, the same idea was independently discovered by Adleman in 1976 and first mentioned in print by Bloniarz~\cite{bloniarz1979} a few years later. The \emph{threshold-2} function, denoted $T_2(x_1,\dots,x_n)$, is the negation of $\amo(x_1,\dots,x_n)$. Adleman showed that there is a monotone boolean circuit for $T_2(x_1,\dots,x_n)$ with $2n + 2 \sqrt{n} + O(\sqrt[4]{n})$ gates. Our construction from \Cref{thm:am1} can naturally be adapted to circuits, yielding the first improvement to Adleman's result:
\begin{restatable}{theorem}{thmcircuit} \label{thm:circuit}
    There is a monotone boolean circuit for $T_2(x_1,\dots,x_n)$ with $2n + \sqrt{2n} + O(\sqrt[3]{n})$ gates.
\end{restatable}
Sergeev~\cite{sergeev} proved that every monotone boolean circuit for $T_2(x_1,\dots,x_n)$ has at least $2n + \sqrt{(2n-4)/3} - 19/6$ gates, so \Cref{thm:circuit} is almost optimal.

A monotone boolean circuit is said to be \emph{single level} if every path from an input to the output goes through at most one $\land$ gate. Interestingly, Sergeev showed that every \emph{single-level} monotone boolean circuit for $T_2(x_1,\dots,x_n)$ has at least $2n + 2 \sqrt{n+11} - 10$ gates.\footnote{In \cite{sergeev}, the bound is stated as $2n + 2 \sqrt{n+27} - 14$. Sergeev told us that the proof contains a mistake that, when corrected, yields the (improved) bound $2n + 2 \sqrt{n+11} - 10$.} Thus, Adleman's construction is essentially optimal for single-level circuits, and a corollary of \Cref{thm:circuit} is that the smallest monotone boolean circuits for $T_2(x_1,\dots,x_n)$ are not single level. This answers a 47-year-old open question from Bloniarz~\cite[p.~158]{bloniarz1979}, and it should be contrasted with a result of Krichevskii~\cite{krichevskii} stating that single-level monotone boolean \emph{formulas} are optimal for $T_2(x_1,\dots,x_n)$. It was a long-standing open problem whether there exists a quadratic boolean function (i.e., a disjunction of cubes of the form $x_i \land x_j$) whose single-level monotone circuit complexity is strictly greater than its monotone circuit complexity; the negation of this statement was sometimes called the \emph{single-level conjecture}. The problem appears to originate with Bloniarz~\cite[p.~158]{bloniarz1979} and was further studied by Lenz and Wegener~\cite{lenz-wegener} and several other authors (see, e.g., \cite{bublitz,mirwald-schnorr,amano-maruoka}). The conjecture was finally disproved by Jukna~\cite{single-level} using a carefully constructed quadratic boolean function. Thus, our results show that, surprisingly, the conjecture already fails for $T_2(x_1,\dots,x_n)$, the simplest quadratic boolean function of all.

\subsection{\texorpdfstring{An unconditional lower bound for $\amo$}{An unconditional lower bound for AMO}}


Unconditional lower bounds on CNF encodings are rare (directly related to circuit lower bounds and thus $\textsf{P} \stackrel{?}{=} \textsf{NP}$~\cite{Aaronson2016}); the only example we are aware of is for the \emph{parity} function, for which Emdin, Kulikov, Mihajlin, and Slezkin~\cite{cnf-symmetric} showed a $3n - 9$ lower bound, to be compared with the $4n - 6$ known upper bound.

Ku{\v c}era, Savick{\'{y}}, and Vorel~\cite{lower-bound} proved that every propagation-complete encoding of the $\amo(x_1, \ldots, x_n)$ constraint requires $2n + \sqrt{n} - 2$ clauses for $n \ge 7$, and that this can be improved to $2n + 2\sqrt{n} - 3$ clauses for $n \ge 9$ for 2-CNF encodings. They asked whether a $2n + \Omega(\sqrt{n})$ lower bound holds for \emph{unit refutation complete} encodings of $\amo(x_1, \ldots, x_n)$, where unit refutation completeness is a weaker version of propagation completeness. \Cref{thm:lower-bound} provides such a bound with no propagation assumptions. 
Our lower bound proof is heavily inspired by Ku{\v c}era, Savick{\'{y}}, and Vorel's proof. They define a \emph{regular form} for encodings of $\amo(x_1, \ldots, x_n)$, which is well-suited for theoretical analysis. The proof strategy is two-pronged. On one hand, if the smallest encoding $\varphi$ of $\amo(x_1, \ldots, x_n)$ is not in regular form, then they show that there is an encoding $\varphi'$ of $\amo(x_1, \ldots, x_{n-1})$ with three fewer clauses, which allows one to conclude by induction. 
On the other hand, if an encoding $\varphi$ is in regular form, then they argue directly that the encoding satisfies their claimed bound. Both prongs of their argument use the propagation completeness assumption. Our contribution is to show that the assumption can be eliminated in both cases. The analysis of encodings in regular form uses a graph-theoretic argument that is inspired by the perspective in \Cref{sec-perspective} (see \Cref{lem-regular-lb}).

The lower bound for $\amo$ easily extends to $\amk$:
\begin{corollary} \label{cor:lower-bound-amk}
    Every encoding of the $\amk(x_1, \ldots, x_n)$ constraint has at least $2(n-k) + \sqrt{n-k+2}$ clauses for $n \ge k + 7$.
\end{corollary}
\begin{proof}
    Let $\varphi$ be an encoding of $\amk(x_1, \ldots, x_n)$, and let $\tau$ be the partial assignment setting $x_1, \dots, x_{k-1}$ to $\top$. Then, note that $\varphi|_\tau$ encodes $\amo(x_k, \dots, x_{n})$, and thus by~\Cref{thm:lower-bound}, it must have at least $2(n-k+1) + \sqrt{n-k+2} - 2$ clauses. Since $\varphi$ has at least as many clauses as $\varphi|_\tau$, this concludes the proof.
\end{proof}

\section{\texorpdfstring{Disjunctive Switching Yields a $2n + o_k(n)$ Encoding for $\amk$}{Disjunctive Switching Yields a 2n + o\_k(n) Encoding for AMK}}
\label{sec:gdpe}

Before the present work, the smallest known encodings for $\amk(x_1,\dots,x_n)$ were of size $(k+1)n + o(n)$ for $k \le 5$~\cite{product-k} and $7n - o(n)$ for $k \ge 6$~\cite{sinz}, and \Cref{cor:lower-bound-amk} gives a lower bound of $2n$ for $k = o(n)$. The encoding of size $(k+1)n + o(n)$ is based on a generalization of Chen's product encoding. In this section, we show that the generalized product encoding can be ``disjunctivized'', yielding the following:
\begin{restatable}{theorem}{amkdisjunctiveproduct} \label{thm:amk-disjunctive-product}
    There is an encoding of the $\amk(x_1,\dots,x_n)$ constraint using $2n + O(kn^{k/(k+1)})$ clauses and $O(kn^{k/(k+1)})$ auxiliary variables.
\end{restatable}
\noindent
For $k=o(\log n / \log \log n)$, this encoding uses $2n + o(n)$ clauses and $o(n)$ auxiliary variables.

We begin by introducing the concept of disjunctive switching through an example. Then, we present the generalized product encoding of Frisch and Giannaros~\cite{product-k} and sketch how disjunctive switching allows us to obtain~\Cref{thm:amk-disjunctive-product} from it.


\subsection{Disjunctive switching} \label{sec-disjunctive-switching}

Let us illustrate the main idea behind disjunctive switching through an example and then offer a more abstract perspective in Appendix~\ref{app:disjunctive-switching}.

Recall that in the grid presentation of Chen's product encoding (\Cref{fig:grid}), having two input variables $x_{i, j}$ and $x_{i', j'}$ assigned to true will either contradict that at most one row is selected or that at most one column is selected. Based on which of these two constraints is violated, one might in hindsight feel as if one of the two implications (a) $x_{i, j} \to r_i$ (b) $x_{i, j} \to c_j$ was ``wasted''. 
Disjunctive switching will use \emph{disjunctive implications} $x_{i, j} \to (r_i \lor c_j)$ and then incorporate some \emph{switching} mechanism to negate whichever variable (i.e., $r_i$ or $c_j$) would have been ``wasted''.


We illustrate the technique with the following problem: we have a $3\times3$ grid of variables $x_{1,1}, \dots, x_{3,3}$, as illustrated in~\Cref{fig:exampleswitch}, and a ``direction'' variable $d$. We wish to encode that if $d$ is true, then at most one column of the grid contains a true $x$ variable, and that if $d$ is false, then at most one row of the grid contains a true $x$ variable. In other words, variable $d$ corresponds to which direction (horizontal or vertical) will be constrained. 
\begin{figure}[h]
    \centering
\begin{tikzpicture}
    \def\sizemult{0.75}
    \newcommand{\drawgrid}[2]{
        \begin{scope}[xshift=#1]
            \foreach \row/\col in {#2} {
                \fill[green, opacity=0.3] (\sizemult*\col-\sizemult*1, \sizemult*3-\sizemult*\row) rectangle (\sizemult*\col, \sizemult*4-\sizemult*\row);
            }
            
            \draw[step=\sizemult*1cm, black, thick] (0,0) grid (\sizemult*3,\sizemult*3);
            
            \foreach \row in {1,2,3} {
                \foreach \col in {1,2,3} {
                    \node at (\sizemult*\col-\sizemult*0.5, \sizemult*3.5-\sizemult*\row) {$x_{\row,\col}$};
                }
            }
        \end{scope}
    }

    \drawgrid{0cm}{1/1, 2/2, 3/3} 

    \def\moveup{0.25}
    \node[] at (1.25, -0.75+\moveup) {$d = \top$ \textcolor{red!80!black}{\ding{55}}};
    \node[] at (1.25, -1.25+\moveup) {$d = \bot$ \textcolor{red!80!black}{\ding{55}}};

      \node[] at (4.25+0.5, -0.75+\moveup) {$d = \top$ \textcolor{red!80!black}{\ding{55}}};
    \node[] at (4.25+0.5, -1.25+\moveup) {$d = \bot$ \textcolor{green!60!black}{\ding{51}}};

      \node[] at (7.25+1, -0.75+\moveup) {$d = \top$ \textcolor{green!60!black}{\ding{51}}};
    \node[] at (7.25+1, -1.25+\moveup) {$d = \bot$ \textcolor{green!60!black}{\ding{51}}};

      \node[] at (10.25+1.5, -0.75+\moveup) {$d = \top$ \textcolor{green!60!black}{\ding{51}}};
    \node[] at (10.25+1.5, -1.25+\moveup) {$d = \bot$ \textcolor{red!80!black}{\ding{55}}};

    \drawgrid{3.5cm}{1/1, 1/2, 1/3}

    \drawgrid{7cm}{2/2}

    \drawgrid{10.5cm}{1/2, 3/2}

    \node[text=green!60!black] at (-0.3, 1.9) {$r_1$};
    \node[text=green!60!black] at (-0.3, 1.9-0.75) {$r_2$};
    \node[text=green!60!black] at (-0.3, 1.9-1.5) {$r_3$};

    \node[text=green!60!black] at (-0.3+3.5, 1.9) {$r_1$};
    \node[] at (-0.3+3.5, 1.9-0.75) {$r_2$};
    \node[] at (-0.3+3.5, 1.9-1.5) {$r_3$};

    \node[] at (-0.3+7.0, 1.9) {$r_1$};
    \node[text=green!60!black] at (-0.3+7.0, 1.9-0.75) {$r_2$};
    \node[] at (-0.3+7.0, 1.9-1.5) {$r_3$};

    \node[text=green!60!black] at (-0.3+10.5, 1.9) {$r_1$};
    \node[] at (-0.3+10.5, 1.9-0.75) {$r_2$};
    \node[text=green!60!black] at (-0.3+10.5, 1.9-1.5) {$r_3$};

    \node[text=green!60!black] at (1*0.75-0.4, 2.5) {$c_1$};
    \node[text=green!60!black] at (2*0.75-0.4, 2.5) {$c_2$};
    \node[text=green!60!black] at (3*0.75-0.4, 2.5) {$c_3$};

    \node[text=green!60!black] at (1*3.5+1*0.75-0.4, 2.5) {$c_1$};
    \node[text=green!60!black] at (1*3.5+2*0.75-0.4, 2.5) {$c_2$};
    \node[text=green!60!black] at (1*3.5+3*0.75-0.4, 2.5) {$c_3$};

    \node[] at (2*3.5+1*0.75-0.4, 2.5) {$c_1$};
    \node[text=green!60!black] at (2*3.5+2*0.75-0.4, 2.5) {$c_2$};
    \node[] at (2*3.5+3*0.75-0.4, 2.5) {$c_3$};

    \node[] at (3*3.5+1*0.75-0.4, 2.5) {$c_1$};
    \node[text=green!60!black] at (3*3.5+2*0.75-0.4, 2.5) {$c_2$};
    \node[] at (3*3.5+3*0.75-0.4, 2.5) {$c_3$};


\end{tikzpicture}
    \caption{Example for the disjunctive switching technique. True input variables are highlighted in green. The \textcolor{green!60!black}{\ding{51}} and \textcolor{red!80!black}{\ding{55}} symbols indicate whether the encoding should be satisfiable for the given $d$.}
    \label{fig:exampleswitch}
\end{figure}
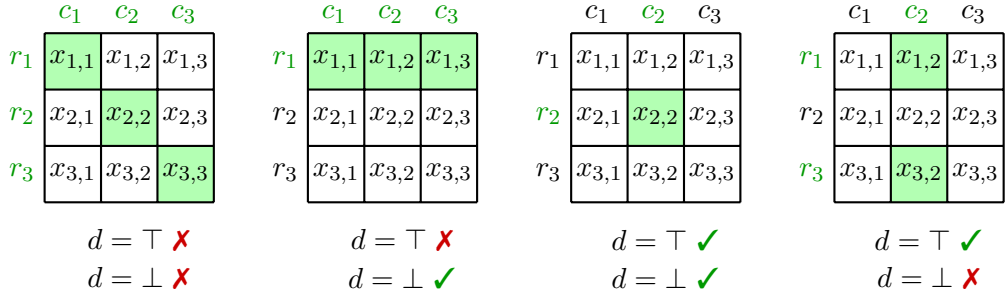
A na\"ive encoding would resemble Chen's product encoding; we start with $2$ clauses per input variable
\begin{equation}\label{eq:18clauses}
    \bigwedge_{i = 1}^3 \bigwedge_{j=1}^3 \lrp{x_{i,j} \rightarrow r_i} \land  \lrp{x_{i,j} \rightarrow c_j},
\end{equation}
and then we enforce $d \rightarrow \amo(c_1, c_2, c_3 )$ and  $\overline{d} \rightarrow \amo(r_1, r_2, r_3)$, both of which consist of 3 clauses.
Thus, this encoding uses a total of $2 \cdot 9 + 2 \cdot 3 = 24$ clauses.

The disjunctive switching paradigm allows us to reduce the number of clauses. We start with the $9$ disjunctive clauses
\begin{equation}\label{eq:disjunctive-imp}
    \bigwedge_{i = 1}^3 \bigwedge_{j=1}^3 \lrp{x_{i, j} \rightarrow (r_i \lor c_j)} \equiv  \bigwedge_{i = 1}^3 \bigwedge_{j=1}^3 \lrp{\overline{x_{i, j}} \lor r_i \lor c_j},
\end{equation}
instead of the $18$ clauses from constraint~\eqref{eq:18clauses}. Now, we add the  $\amo$ constraints
\begin{equation}\label{eq:rowamo}
    (\overline{r_1} \lor \overline{r_2}) \land  (\overline{r_1} \lor \overline{r_3}) \land  (\overline{r_2} \lor \overline{r_3}) \land  (\overline{c_1} \lor \overline{c_2}) \land  (\overline{c_1} \lor \overline{c_3}) \land  (\overline{c_2} \lor \overline{c_3}),
\end{equation}
and we complete our encoding by adding the following clauses:
\begin{equation}\label{eq:turning-off}    
    (d \to \overline{r_1}) \land (d \to \overline{r_2}) \land (d \to \overline{r_3}) \land     (\overline{d} \to \overline{c_1}) \land (\overline{d} \to \overline{c_2}) \land (\overline{d} \to \overline{c_3}). 
\end{equation}
Intuitively, if $d$ is true, then it ``turns off'' all the $r_i$ variables, and thus (i) their $\amo$ constraint in~\eqref{eq:rowamo} is trivially satisfied, and (ii) the disjunctive implications in constraint~\eqref{eq:disjunctive-imp} reduce to $x_{i, j} \to c_j$, and thus \eqref{eq:rowamo} enforces an at-most-one-column constraint. Naturally, for the other direction, $\overline{d}$ ``turns off'' the $c_j$ variables through~\eqref{eq:turning-off}, and thus~\eqref{eq:disjunctive-imp} together with~\eqref{eq:rowamo} end up enforcing an at-most-one-row constraint.
This disjunctive-switching encoding uses $9 + 6 + 6 = 21$ clauses, saving 3 from the na\"ive one. 
For an $n \times n$ grid, the na\"ive encoding would have $2n^2 + O(n)$ clauses, whereas the disjunctive-switching one would only have $n^2 + O(n)$.
Next, we will see that the same principle used in this encoding can serve as the basis of a compact encoding for $\amk$.

\subsection{The generalized product encoding}

\begin{figure}[b]
    \begin{subfigure}[t]{0.48\linewidth}
    \centering
    \scalebox{0.7}{%
\begin{tikzpicture}[
    3d view={-145}{32},
    line join=round,
    line cap=round
]
    \def\N{4}
    \def\S{4}
    \def\gap{0.02}
    \def\lw{0.35pt}
    \def\cellopacity{0.82}

    \newcommand{\TopCell}[3]{%
        \pgfmathsetmacro\xa{#2-1+\gap}
        \pgfmathsetmacro\xb{#2-\gap}
        \pgfmathsetmacro\ya{#3-1+\gap}
        \pgfmathsetmacro\yb{#3-\gap}
        \path[fill=#1,fill opacity=\cellopacity,draw=white,line width=\lw]
            (\xa,\ya,\S)--(\xb,\ya,\S)--(\xb,\yb,\S)--(\xa,\yb,\S)--cycle;
    }

    \newcommand{\XFaceCell}[3]{%
        \pgfmathsetmacro\ya{#2-1+\gap}
        \pgfmathsetmacro\yb{#2-\gap}
        \pgfmathsetmacro\za{#3-1+\gap}
        \pgfmathsetmacro\zb{#3-\gap}
        \path[fill=#1,fill opacity=\cellopacity,draw=white,line width=\lw]
            (0,\ya,\za)--(0,\yb,\za)--(0,\yb,\zb)--(0,\ya,\zb)--cycle;
    }

    \newcommand{\YFaceCell}[3]{%
        \pgfmathsetmacro\xa{#2-1+\gap}
        \pgfmathsetmacro\xb{#2-\gap}
        \pgfmathsetmacro\za{#3-1+\gap}
        \pgfmathsetmacro\zb{#3-\gap}
        \path[fill=#1,fill opacity=\cellopacity,draw=white,line width=\lw]
            (\xa,\S,\za)--(\xb,\S,\za)--(\xb,\S,\zb)--(\xa,\S,\zb)--cycle;
    }

    \newcommand{\HighlightCubie}[4]{%
        \pgfmathtruncatemacro{\cx}{#2}
        \pgfmathtruncatemacro{\cy}{#3}
        \pgfmathtruncatemacro{\cz}{#4}
        \ifnum\cz=\N
            \TopCell{#1}{\cx}{\cy}
        \fi
        \ifnum\cx=1
            \XFaceCell{#1}{\cy}{\cz}
        \fi
        \ifnum\cy=\N
            \YFaceCell{#1}{\cx}{\cz}
        \fi
    }

    \fill[black, opacity=0.08]
        (0.45,0.12,0) -- (4.25,0.12,0) -- (4.88,3.34,0) -- (1.08,3.34,0) -- cycle;
    \fill[black, opacity=0.055]
        (0.28,-0.02,0) -- (4.42,-0.02,0) -- (5.05,3.50,0) -- (0.95,3.50,0) -- cycle;
    \fill[black, opacity=0.035]
        (0.12,-0.18,0) -- (4.58,-0.18,0) -- (5.25,3.66,0) -- (0.84,3.66,0) -- cycle;
    \fill[black, opacity=0.02]
        (-0.02,-0.34,0) -- (4.72,-0.34,0) -- (5.42,3.84,0) -- (0.72,3.84,0) -- cycle;

    \draw[gray!70, line width=0.65pt, opacity=0.68] (0,0,0) -- (\S,0,0) -- (\S,\S,0) -- (0,\S,0) -- cycle;
    \draw[gray!70, line width=0.65pt, opacity=0.68] (0,0,\S) -- (\S,0,\S) -- (\S,\S,\S) -- (0,\S,\S) -- cycle;
    \draw[gray!70, line width=0.65pt, opacity=0.68] (0,0,0) -- (0,0,\S);
    \draw[gray!70, line width=0.65pt, opacity=0.68] (\S,0,0) -- (\S,0,\S);
    \draw[gray!70, line width=0.65pt, opacity=0.68] (\S,\S,0) -- (\S,\S,\S);
    \draw[gray!70, line width=0.65pt, opacity=0.68] (0,\S,0) -- (0,\S,\S);

    \foreach \x in {1,...,\N}{
        \foreach \z in {1,...,\N}{
            \YFaceCell{gray!24}{\x}{\z}
        }
    }
    \foreach \x in {1,...,\N}{
        \foreach \y in {1,...,\N}{
            \TopCell{gray!20}{\x}{\y}
        }
    }
    \foreach \y in {1,...,\N}{
        \foreach \z in {1,...,\N}{
            \XFaceCell{gray!16}{\y}{\z}
        }
    }

    \begin{scope}[transparency group]
        \def\cellopacity{0.30}
        \foreach \y in {1,...,\N}{
            \HighlightCubie{blue!70}{1}{\y}{4}
        }
        \def\cellopacity{0.30}
        \foreach \z in {1,...,\N}{
            \HighlightCubie{green!62!black}{1}{2}{\z}
        }
        \def\cellopacity{0.35}
        \foreach \x in {1,...,\N}{
            \HighlightCubie{red!82}{\x}{2}{4}
        }
    \end{scope}
    \def\cellopacity{0.82}

    \path[draw=black, line width=0.9pt]
        (0+\gap,1+\gap,\S)--(1-\gap,1+\gap,\S)--(1-\gap,2-\gap,\S)--(0+\gap,2-\gap,\S)--cycle;
    \path[draw=black, line width=0.9pt]
        (0,1+\gap,3+\gap)--(0,2-\gap,3+\gap)--(0,2-\gap,4-\gap)--(0,1+\gap,4-\gap)--cycle;

    \coordinate (axisO) at (-1.25, +0.30,0.25);
    \draw[-{Stealth[length=1.8mm]}, line width=0.55pt, gray!70!black, opacity=0.8]
        (axisO) -- ++(0.70,0,0) node[font=\scriptsize, anchor=south] {$x$};
    \draw[-{Stealth[length=1.8mm]}, line width=0.55pt, gray!70!black, opacity=0.8]
        (axisO) -- ++(0,0.70,0) node[font=\scriptsize, anchor=east] {$y$};
    \draw[-{Stealth[length=1.8mm]}, line width=0.55pt, gray!70!black, opacity=0.8]
        (axisO) -- ++(0,0,0.70) node[font=\scriptsize, anchor=south] {$z$};

    \node[font=\scriptsize, anchor=west] at (1.0,1.55,\S-0.15) {$p_{1,2,4}$};
    \node[font=\scriptsize, text=red!70!black, anchor=west] at (\S+0.6,2,\S+0.3) {$r^{x}_{2,4}$};
    \node[font=\scriptsize, text=blue!70!black, anchor=south] at (0.2,-0.2,\S) {$r^{y}_{1,4}$};
    \node[font=\scriptsize, text=green!50!black, anchor=east] at (-0.1,0.8,-0.7) {$r^{z}_{1,2}$};
\end{tikzpicture}%
}
    \caption{Rods implied by $p_{1,2,4}$: $r^x_{2,4}$, $r^y_{1,4}$, and $r^z_{1,2}$}
    \label{fig:gen-prod}
    \end{subfigure}
    \hfill
    \begin{subfigure}[t]{0.48\linewidth}
        \centering
        \scalebox{0.7}{%
\begin{tikzpicture}[
    3d view={-145}{32},
    line join=round,
    line cap=round
]
    \def\N{4}
    \def\S{4}
    \def\gap{0.02}
    \def\lw{0.35pt}
    \def\cellopacity{0.82}

    \colorlet{pA}{red!82}
    \colorlet{pB}{blue!70}
    \colorlet{pC}{green!62!black}

    \newcommand{\TopCell}[3]{%
        \pgfmathsetmacro\xa{#2-1+\gap}
        \pgfmathsetmacro\xb{#2-\gap}
        \pgfmathsetmacro\ya{#3-1+\gap}
        \pgfmathsetmacro\yb{#3-\gap}
        \path[fill=#1,fill opacity=\cellopacity,draw=white,line width=\lw]
            (\xa,\ya,\S)--(\xb,\ya,\S)--(\xb,\yb,\S)--(\xa,\yb,\S)--cycle;
    }

    \newcommand{\XFaceCell}[3]{%
        \pgfmathsetmacro\ya{#2-1+\gap}
        \pgfmathsetmacro\yb{#2-\gap}
        \pgfmathsetmacro\za{#3-1+\gap}
        \pgfmathsetmacro\zb{#3-\gap}
        \path[fill=#1,fill opacity=\cellopacity,draw=white,line width=\lw]
            (0,\ya,\za)--(0,\yb,\za)--(0,\yb,\zb)--(0,\ya,\zb)--cycle;
    }

    \newcommand{\YFaceCell}[3]{%
        \pgfmathsetmacro\xa{#2-1+\gap}
        \pgfmathsetmacro\xb{#2-\gap}
        \pgfmathsetmacro\za{#3-1+\gap}
        \pgfmathsetmacro\zb{#3-\gap}
        \path[fill=#1,fill opacity=\cellopacity,draw=white,line width=\lw]
            (\xa,\S,\za)--(\xb,\S,\za)--(\xb,\S,\zb)--(\xa,\S,\zb)--cycle;
    }

    \newcommand{\FaceCubie}[3]{%
        \pgfmathsetmacro\xa{#2-1+\gap}
        \pgfmathsetmacro\xb{#2-\gap}
        \pgfmathsetmacro\za{#3-1+\gap}
        \pgfmathsetmacro\zb{#3-\gap}
        \pgfmathsetmacro\yin{\S-1}
        \path[fill=#1, fill opacity=0.36, draw=#1!75!black, draw opacity=0.62, line width=0.42pt]
            (\xa,\S,\za)--(\xb,\S,\za)--(\xb,\S,\zb)--(\xa,\S,\zb)--cycle;
        \path[fill=#1!72!black, fill opacity=0.16, draw=none]
            (\xa,\S,\za)--(\xa,\yin,\za)--(\xa,\yin,\zb)--(\xa,\S,\zb)--cycle;
        \path[fill=#1!45!white, fill opacity=0.16, draw=none]
            (\xa,\S,\zb)--(\xb,\S,\zb)--(\xb,\yin,\zb)--(\xa,\yin,\zb)--cycle;
        \draw[#1!70!black, opacity=0.42, line width=0.32pt]
            (\xa,\yin,\za)--(\xb,\yin,\za)--(\xb,\yin,\zb)--(\xa,\yin,\zb)--cycle;
        \draw[#1!70!black, opacity=0.40, line width=0.30pt] (\xa,\S,\za)--(\xa,\yin,\za);
        \draw[#1!70!black, opacity=0.40, line width=0.30pt] (\xa,\S,\zb)--(\xa,\yin,\zb);
        \draw[#1!70!black, opacity=0.40, line width=0.30pt] (\xb,\S,\zb)--(\xb,\yin,\zb);
    }

    \newcommand{\ProjectionPrism}[3]{%
        \pgfmathsetmacro\xa{#2-1+\gap}
        \pgfmathsetmacro\xb{#2-\gap}
        \pgfmathsetmacro\za{#3-1+\gap}
        \pgfmathsetmacro\zb{#3-\gap}
        \pgfmathsetmacro\ya{\S}
        \pgfmathsetmacro\yb{\S-1.00}
        \path[fill=#1, fill opacity=0.08, draw=#1!70!black, draw opacity=0.24, line width=0.34pt]
            (\xa,\yb,\za)--(\xb,\yb,\za)--(\xb,\yb,\zb)--(\xa,\yb,\zb)--cycle;
        \path[fill=#1, fill opacity=0.06, draw=#1!70!black, draw opacity=0.20, line width=0.30pt]
            (\xa,\ya,\za)--(\xa,\yb,\za)--(\xa,\yb,\zb)--(\xa,\ya,\zb)--cycle;
        \path[fill=#1, fill opacity=0.06, draw=#1!70!black, draw opacity=0.20, line width=0.30pt]
            (\xa,\ya,\zb)--(\xb,\ya,\zb)--(\xb,\yb,\zb)--(\xa,\yb,\zb)--cycle;
    }

    \newcommand{\YFaceBorder}[2]{%
        \pgfmathsetmacro\xa{#1-1+\gap}
        \pgfmathsetmacro\xb{#1-\gap}
        \pgfmathsetmacro\za{#2-1+\gap}
        \pgfmathsetmacro\zb{#2-\gap}
        \path[draw=black!85, line width=0.8pt, dashed]
            (\xa,\S,\za)--(\xb,\S,\za)--(\xb,\S,\zb)--(\xa,\S,\zb)--cycle;
    }

    \newcommand{\TopProjCell}[3]{%
        \pgfmathsetmacro\xa{#2-1+\gap}
        \pgfmathsetmacro\xb{#2-\gap}
        \pgfmathsetmacro\ya{#3-1+\gap}
        \pgfmathsetmacro\yb{#3-\gap}
        \path[fill=#1, fill opacity=0.30, draw=white, line width=\lw]
            (\xa,\ya,\S)--(\xb,\ya,\S)--(\xb,\yb,\S)--(\xa,\yb,\S)--cycle;
    }

    \newcommand{\XProjCell}[3]{%
        \pgfmathsetmacro\ya{#2-1+\gap}
        \pgfmathsetmacro\yb{#2-\gap}
        \pgfmathsetmacro\za{#3-1+\gap}
        \pgfmathsetmacro\zb{#3-\gap}
        \path[fill=#1, fill opacity=0.30, draw=white, line width=\lw]
            (0,\ya,\za)--(0,\yb,\za)--(0,\yb,\zb)--(0,\ya,\zb)--cycle;
    }

    \newcommand{\TopProjSplitCell}[4]{%
        \pgfmathsetmacro\xa{#3-1+\gap}
        \pgfmathsetmacro\xb{#3-\gap}
        \pgfmathsetmacro\ya{#4-1+\gap}
        \pgfmathsetmacro\yb{#4-\gap}
        \path[fill=#1, fill opacity=0.30, draw=none]
            (\xa,\ya,\S)--(\xb,\ya,\S)--(\xb,\yb,\S)--cycle;
        \path[fill=#2, fill opacity=0.30, draw=none]
            (\xa,\ya,\S)--(\xb,\yb,\S)--(\xa,\yb,\S)--cycle;
        \path[draw=white, line width=\lw]
            (\xa,\ya,\S)--(\xb,\ya,\S)--(\xb,\yb,\S)--(\xa,\yb,\S)--cycle;
        \path[draw=gray!55, line width=0.32pt, opacity=0.75]
            (\xa,\ya,\S)--(\xb,\yb,\S);
    }

    \newcommand{\XProjSplitCell}[4]{%
        \pgfmathsetmacro\ya{#3-1+\gap}
        \pgfmathsetmacro\yb{#3-\gap}
        \pgfmathsetmacro\za{#4-1+\gap}
        \pgfmathsetmacro\zb{#4-\gap}
        \path[fill=#1, fill opacity=0.30, draw=none]
            (0,\ya,\za)--(0,\yb,\za)--(0,\yb,\zb)--cycle;
        \path[fill=#2, fill opacity=0.30, draw=none]
            (0,\ya,\za)--(0,\yb,\zb)--(0,\ya,\zb)--cycle;
        \path[draw=white, line width=\lw]
            (0,\ya,\za)--(0,\yb,\za)--(0,\yb,\zb)--(0,\ya,\zb)--cycle;
        \path[draw=gray!55, line width=0.32pt, opacity=0.75]
            (0,\ya,\za)--(0,\yb,\zb);
    }

    \fill[black, opacity=0.08]
        (0.45,0.12,0) -- (4.25,0.12,0) -- (4.88,3.34,0) -- (1.08,3.34,0) -- cycle;
    \fill[black, opacity=0.055]
        (0.28,-0.02,0) -- (4.42,-0.02,0) -- (5.05,3.50,0) -- (0.95,3.50,0) -- cycle;
    \fill[black, opacity=0.035]
        (0.12,-0.18,0) -- (4.58,-0.18,0) -- (5.25,3.66,0) -- (0.84,3.66,0) -- cycle;
    \fill[black, opacity=0.02]
        (-0.02,-0.34,0) -- (4.72,-0.34,0) -- (5.42,3.84,0) -- (0.72,3.84,0) -- cycle;

    \draw[gray!70, line width=0.65pt, opacity=0.68] (0,0,0) -- (\S,0,0) -- (\S,\S,0) -- (0,\S,0) -- cycle;
    \draw[gray!70, line width=0.65pt, opacity=0.68] (0,0,\S) -- (\S,0,\S) -- (\S,\S,\S) -- (0,\S,\S) -- cycle;
    \draw[gray!70, line width=0.65pt, opacity=0.68] (0,0,0) -- (0,0,\S);
    \draw[gray!70, line width=0.65pt, opacity=0.68] (\S,0,0) -- (\S,0,\S);
    \draw[gray!70, line width=0.65pt, opacity=0.68] (\S,\S,0) -- (\S,\S,\S);
    \draw[gray!70, line width=0.65pt, opacity=0.68] (0,\S,0) -- (0,\S,\S);

    \foreach \x in {1,...,\N}{
        \foreach \z in {1,...,\N}{
            \YFaceCell{gray!24}{\x}{\z}
        }
    }
    \foreach \x in {1,...,\N}{
        \foreach \y in {1,...,\N}{
            \TopCell{gray!20}{\x}{\y}
        }
    }
    \foreach \y in {1,...,\N}{
        \foreach \z in {1,...,\N}{
            \XFaceCell{gray!16}{\y}{\z}
        }
    }


    \FaceCubie{pA}{2}{3}
    \FaceCubie{pB}{4}{1}
    \FaceCubie{pC}{4}{3}
    \YFaceBorder{2}{3}
    \YFaceBorder{4}{1}
    \YFaceBorder{4}{3}

    \TopProjCell{pA}{2}{4}
    \TopProjSplitCell{pB}{pC}{4}{4}

    \XProjCell{pB}{4}{1}
    \XProjSplitCell{pA}{pC}{4}{3}


    \node[] at (1.5, \S, 2.50) {\scriptsize $r^y_{2,3}$};

     \node[] at (3.5, \S, 2.50) {\scriptsize $r^y_{4,3}$};

     \node[] at (3.5, \S, 0.50) {\scriptsize $r^y_{4,1}$};

     \node[] at (-0.45, \S+0.1, 1) {\scriptsize $r^x_{4,1}$};

          \node[] at (-0.45, \S-1.5, 2.5) {\large $A_{y, z}$};

    \node[] at (-0.45, \S+0.1, 3) {\scriptsize $r^x_{4,3}$};

    \node[] at (3.5, \S-0.5, 4) {\scriptsize $r^z_{4,4}$};

    \node[] at (1.5, \S-0.5, 4) {\scriptsize $r^z_{2,4}$};

     \node[] at (1.75, \S-1.5, 4.5) {\large $A_{x,y}$};

       \node[] at (1.85, \S-1.25, 0.25) {\large $A_{x,z}$};


    \coordinate (axisO) at (-1.25,+0.30,0.25);
    \draw[-{Stealth[length=1.8mm]}, line width=0.55pt, gray!70!black, opacity=0.8]
        (axisO) -- ++(0.70,0,0) node[font=\scriptsize, anchor=south] {$x$};
    \draw[-{Stealth[length=1.8mm]}, line width=0.55pt, gray!70!black, opacity=0.8]
        (axisO) -- ++(0,0.70,0) node[font=\scriptsize, anchor=east] {$y$};
    \draw[-{Stealth[length=1.8mm]}, line width=0.55pt, gray!70!black, opacity=0.8]
        (axisO) -- ++(0,0,0.70) node[font=\scriptsize, anchor=south] {$z$};

    \node[font=\scriptsize, text=black, anchor=south west] at (1.95,\S,2.95) {$p_{2,4,3}$};
    \node[font=\scriptsize, text=black, anchor=north east] at (4.0,\S,1.0) {$p_{4,4,1}$};
    \node[font=\scriptsize, text=black, anchor=south west] at (3.95,\S,2.95) {$p_{4,4,3}$};
\end{tikzpicture}%
}
        \caption{Three input variables and rods visible from different faces}
        \label{fig:gen-prod2}
    \end{subfigure}
        \caption{Example of the generalized product encoding for $k=2$}
\end{figure}
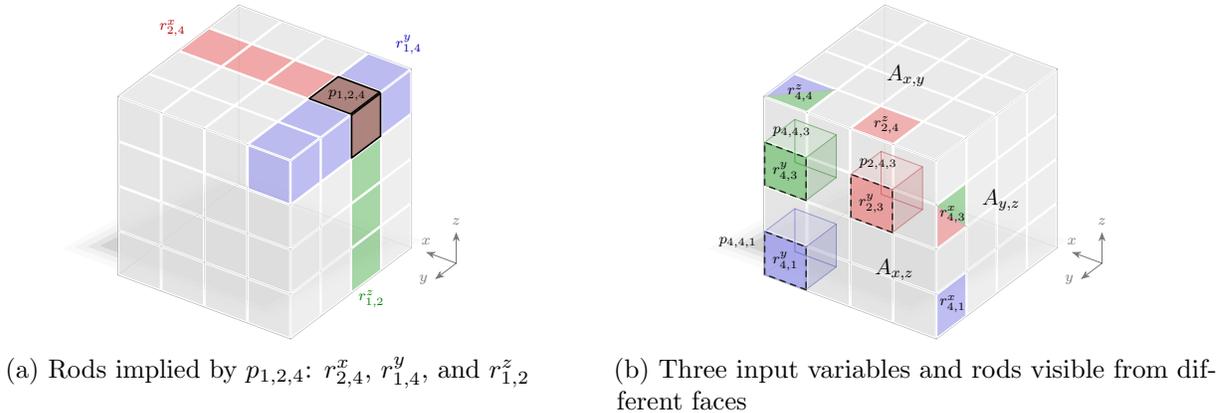

Frisch and Giannaros proposed a nice generalization of Chen's product encoding~\cite{product-k}; rather than imagining the input variables in a two-dimensional grid as in \Cref{fig:grid}, we instead imagine them in a $(k+1)$-dimensional grid, and use the higher-dimensional analogues of rows and columns to recursively encode $\amk$. Let us detail the particular case of $k=2$ and present the general case in Appendix~\ref{appendix:dgpe}.

For $k=2$, the input variables $p_{x,y,z}$ take place in a three-dimensional grid (for $n$ variables, it will be a $\ceil{\sqrt[3]{n}} \times \ceil{\sqrt[3]{n}} \times \ceil{\sqrt[3]{n}}$ grid), and each variable implies the three ``rods'' it belongs to, as e.g.,
\(
    (p_{1, 2, 4} \to r^x_{2, 4}) \land (p_{1,2,4} \to r^y_{1,4}) \land (p_{1,2,4} \to r^z_{1,2}),
\)
depicted in~\Cref{fig:gen-prod}.

Then, we let $A_{y,z} = \left\{r^x_{y, z} \mid 1 \leq y,z \leq \ceil{\sqrt[3]{n}} \right\}$ be the rods visible from the $yz$ face of the grid, and similarly let $A_{x,z}$ and $A_{x,y}$ be the rods visible from the $xz$ and $xy$ faces.
The encoding is then completed by recursively enforcing that from every face at most two rods are visible: $\amt\left(A_{y,z}\right) \land \amt(A_{x,z})  \land \amt(A_{x,y})$.

Clearly, if at most two input variables are true, then the encoding is satisfiable. Less obviously, if at least three input variables are true, then there will be some face $yz$, $xz$, or $xy$ from which at least three implied rods are visible, making the encoding unsatisfiable. This is directly analogous to the key insight justifying the product encoding, namely that if at least two input variables are true, then there is some \emph{side} (i.e., $x$-axis or $y$-axis) from which at least two implied rows/columns are visible.

Note that there are $3n$ \emph{projection} clauses, each stating that an input variable $p_{x,y,z}$ implies a rod variable. Moreover, there are $O(n^{2/3})$ rods $A_{y,z}$, $A_{x,z}$, and $A_{x,y}$, from where the recursive $\amt$ constraints use only $O(n^{2/3})$ clauses. Hence, the encoding is of size $3n + O(n^{2/3})$.

\subsection{The disjunctive generalized product encoding}

The bottleneck of the generalized product encoding is the projection clauses, which consist of $3n$ clauses for $k=2$. Using disjunctive switching, we can reduce this burden to only $2n$ clauses, yielding a $2n+o(n)$ encoding for fixed $k$. Here, we describe the encoding for $k=2$, leaving a full description to the appendix.

Recall that the correctness of the previous encoding relied on the following fact: if at least three input variables are true, then there is some face $yz$, $xz$, or $xy$ from which at least three implied rods are visible. Call such a face a \emph{witnessing} face. If we knew a priori which face is witnessing, then we could readily apply disjunctive switching to encode the projections using only $n$ clauses. Our main insight for this encoding is that, while we do not know which face is witnessing a priori, we can determine it after projecting onto the first face $yz$. For example, in \Cref{fig:gen-prod2}, the visible rods from the $yz$-face are $r^x_{4,1}$ and $r^x_{4,3}$, and since these rods share a $y$-coordinate, projecting onto the $xy$-face is a bad idea: the $r^x_{4,1}$ and $r^x_{4,3}$ variables do not carry information of the $x$-coordinates of the input variables implying them, so if those $x$-coordinates turned out to be equal (as is the case for $p_{4,4,1}$ and $p_{4,4,3}$ in \Cref{fig:gen-prod2}), their $xy$-projections would coincide ($r^z_{4,4}$ in \Cref{fig:gen-prod2}).
Therefore, we conclude that $xz$ must be the witnessing face if there is one.

We are now ready for a formal description of the encoding. First, we include projection clauses, direct ones for the $yz$-face and ``disjunctivized'' ones for the other two:
\begin{equation} \label{eq-disj-proj}
    \bigwedge_{x,y,z \in \left[\ceil{\sqrt[3]{n}}\right]} \left(\overline{p_{x,y,z}} \lor r^x_{y,z}\right) \land \bigwedge_{x,y,z \in \left[\ceil{\sqrt[3]{n}}\right]} \left(\overline{p_{x,y,z}} \lor r^y_{x,z} \lor r^z_{x,y}\right).
\end{equation}
Then, we impose the following constraints on the $xz$- and $xy$-faces: $\amt(A_{x,z}) \land \amt(A_{x,y})$; it will turn out to be superfluous to impose $\amt(A_{y,z})$.

Next, we introduce an auxiliary variable $w$, whose intended semantics is that $xz$ is a witnessing face (if there is one). 
By the discussion above, if there are two visible rods $r^x_{y,z_1}$ and $r^x_{y,z_2}$, then $w$ must be true. Thus, we include the following constraint:
\[
    \bigwedge_{y \in \left[\ceil{\sqrt[3]{n}}\right]} \left(w \lor \amo\left(\left\{r^x_{y,z} \mid z \in  \left[\ceil{\sqrt[3]{n}}\right]\right\}\right)\right), 
\]
and analogously for the symmetric case,
\[
 \bigwedge_{z \in \left[\ceil{\sqrt[3]{n}}\right]} \left(\overline{w} \lor \amo\left(\left\{r^x_{y,z} \mid y \in  \left[\ceil{\sqrt[3]{n}}\right]\right\}\right)\right).
\]
Finally, based on $w$, we constrain the projection variables for the $xz$- and $xy$-faces (cf.~\eqref{eq:turning-off}):
\[
     \bigwedge_{x,y \in \left[\ceil{\sqrt[3]{n}}\right]} \left(\overline{r^z_{x,y}} \lor \overline{w}\right) \land \bigwedge_{x,z \in \left[\ceil{\sqrt[3]{n}}\right]} \left(\overline{r^y_{x,z}} \lor w\right).
\]
This completes the description of the disjunctive generalized product encoding for $k=2$. There are $2n$ projection clauses in \eqref{eq-disj-proj}, and all other constraints use $O(n^{2/3})$ clauses. Thus, the total encoding size is $2n + O(n^{2/3})$.

\section{\texorpdfstring{Grid Compression Encodings for $\amk$}{Grid Compression Encodings for AMK}} \label{sec-grid-compression}
The disjunctive generalized product encoding for $\amk$ has $2n + o(n)$ clauses for $k = o(\tfrac{\log n}{\log \log n})$. In this section, we introduce a technique called \emph{grid compression}, which allows us to give an encoding for $\amk$ with $2n + o(n)$ clauses for $k = o(\sqrt[3]{n})$. We start by describing an encoding with $4n + o(n)$ clauses for any $k = o(n)$. Then, we show how a ``disjunctivized'' version of this yields a $2n + o(n)$ encoding for $k = o(\sqrt[3]{n})$.

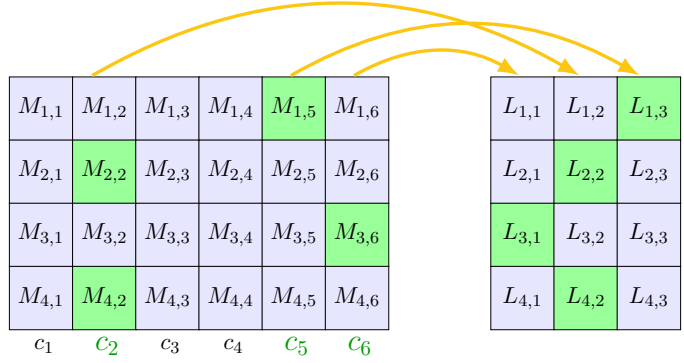
\begin{figure}
\centering
\scalebox{0.8}{
\begin{tikzpicture}[ basenode/.style={circle, fill=blue!60!gray, text=white, inner sep=1pt}, auxnode/.style={circle, fill=green!60!gray, text=white, inner sep=1pt}]

\foreach \i in {1,2,3,4,5,6} {
  \foreach \j in {1,2,3,4} {
    \def\deepblue{blue!10!white}

    \pgfmathsetmacro\il{int(\i)}
         \pgfmathsetmacro\jl{int(\j)}
      \ifnum\il=2\relax
        \ifnum\jl=2\relax
          \def\deepblue{green!40!white}
         \fi
        \ifnum\jl=4\relax
          \def\deepblue{green!40!white}
         \fi
     \fi
     \ifnum\il=5\relax
        \ifnum\jl=1\relax
          \def\deepblue{green!40!white}
         \fi
          
     \fi

      \ifnum\il=6\relax
        \ifnum\jl=3\relax
          \def\deepblue{green!40!white}
         \fi
     \fi

   \node[draw, fill=\deepblue, minimum size=30pt] (b\i\j) at (1.05*\i, -1.05*\j) {$M_{\j, \i}$};

   }

}


\node [] (c1) at (1.1, -5) {$c_1$};
\node [] (c2) at (2.15, -5) {\Large \textcolor{green!60!black}{$c_2$}};
\node [] (c3) at (3.2, -5) {$c_3$};
\node [] (c4) at (4.25, -5) {$c_4$};
\node [] (c5) at (5.3, -5) {\Large \textcolor{green!60!black}{$c_5$}};
\node [] (c6) at (6.35, -5) {\Large \textcolor{green!60!black}{$c_6$}};


\foreach \i in {1, 2, 3} {
  \foreach \j in {1,2,3,4} {
    \def\deepblue{blue!10!white}

    \pgfmathsetmacro\il{int(\i)}

\pgfmathsetmacro\id{int(\i / 2)}
         \pgfmathsetmacro\jl{int(\j)}
      \ifnum\il=2\relax
        \ifnum\jl=2\relax
          \def\deepblue{green!40!white}
         \fi
        \ifnum\jl=4\relax
          \def\deepblue{green!40!white}
         \fi
     \fi
     \ifnum\il=3\relax
        \ifnum\jl=1\relax
          \def\deepblue{green!40!white}
         \fi
          
     \fi

      \ifnum\il=1\relax
        \ifnum\jl=3\relax
          \def\deepblue{green!40!white}
         \fi
     \fi

   \node[draw, fill=\deepblue, minimum size=30pt] (b\i\j) at (8 + 1.05*\i, -1.05*\j) {$L_{\j, \i}$};

   }
}

\draw[ultra thick, lipicsYellow] (1.9, -0.5) edge[-Latex, in=150, out=30] (10, -0.5);
\draw[ultra thick, lipicsYellow] (6.25, -0.5) edge[-Latex, in=150, out=30] (9, -0.5);
\draw[ultra thick, lipicsYellow] (5.2, -0.5) edge[-Latex, in=150, out=30] (11, -0.5);



\end{tikzpicture}
}
\caption{Illustration of the grid compression framework}\label{fig:compression}
\end{figure}

\subsection{General framework}\label{subsec:gridgeneral}
We arrange the input variables $x_1, \ldots, x_n$ in a $\lceil n / m \rceil \times m$ grid $M$, where $m$ is a parameter satisfying $k = o(m)$ and $m = o(n)$ that will be specified later for each encoding. If at most $k$ input variables are set to true, then at most $k$ columns of the grid contain input variables set to true (call these columns \textit{occupied}). Grid compression works by \emph{compressing} $M$ into a $\lceil n / m \rceil \times \ell$ grid $L$---where $\ell$ is a parameter satisfying $k = o(\ell)$ and $\ell = o(m)$---which will contain a copy of each occupied column in $M$, although potentially in a different order (see \Cref{fig:compression}). To that end, we will have auxiliary variables $L_{i, j}$ for $i \in [\ceil{n / m}]$ and $j \in [\ell]$.

To implement grid compression, we will employ intermediate constructions that assign the indices of occupied columns in $M$ to indices of columns in $L$. Then, if column index $j$ of $M$ is mapped to column index $p$ of $L$, we will have another set of constraints that ensure that $M_{:,j} = L_{:,p}$. We then enforce
\begin{equation}
    \amk(\{L_{i, j} \mid (i, j) \in [\ceil{n/m}] \times [\ell]\}) \label{eq:atmostkinL}
\end{equation}
using $O(\ceil{n/m} \times \ell) = o(n)$ clauses through, e.g., the parallel counter encoding~\cite{sinz}.

For each column $j \in [m]$, we will have an auxiliary variable $c_j$ representing whether that column of $M$ is occupied. Correspondingly, we will add clauses of the form
\begin{equation}
\bigwedge_{(i, j) \in [\ceil{n / m}] \times [m]} (\overline{M_{i, j}} \lor c_j). \label{eq:varimpliesitscolinM}
\end{equation}

While both encodings we present next include constraints~\eqref{eq:atmostkinL} and~\eqref{eq:varimpliesitscolinM},
the specific implementation of the copying from $M$ to $L$ depends on which variant of the encoding we employ. The commonality is that, for each $j \in [m]$, we associate some \textit{small} set $H_j \subset [\ell]$, which will represent the set of columns in $L$ that \textit{may} store a copy of the $j$th column of $M$. 
At a high level, our encodings need to use small sets $H_j$ to save clauses, and yet these small sets need to be chosen carefully so that it is still possible to choose a different target column from each $H_j$ such that column $j$ is occupied; this is a traditional problem in hashing, for which our two encodings use different solutions.

\subsection{Grid compression encoding}
\newcommand{\cp}{\mathsf{copy}}

We set $m = \Theta(\sqrt[3]{kn^2})$ and $\ell = \Theta(\sqrt[3]{k^2n})$ for suitable constants. We will have $|H_j| = 3$ for every $j \in [m]$, and the construction of these sets will be specified later. We introduce auxiliary variables $\cp_{j,p}$, representing that $M_{:,j}$ is to be copied to $L_{:,p}$, whose semantics are enforced by the following clauses:
\begin{equation}
     \bigwedge_{(i, j) \in [\ceil{n/m}] \times [m]} \; \bigwedge_{p \in H_j} \lrp{\overline{M_{i, j}} \lor \overline{\cp_{j, p}} \lor L_{i, p}}.\label{eq:colcopyconjunctive}
\end{equation}
We next enforce that each occupied column $j$ of $M$ is copied to some column $p$ of $L$ with $p \in H_j$:
\begin{equation}
    \bigwedge_{j \in [m]} \lrp{\overline{c_j} \lor \bigvee_{p \in H_j} \cp_{j, p}}.\label{eq:setpossiblecopycolsconjunctive}
\end{equation}
Finally, we enforce that no two columns of $M$ can be copied to the same column of $L$:
\begin{equation}
    \bigwedge_{p \in [\ell]} \amo(\{\cp_{j, p} \mid j \in [m] \text{ such that } p \in H_j\}).\label{eq:columnmatching}
\end{equation}

To make the encoding correct, we must choose our collection of sets $\mathcal{H} := \{ H_j \mid j \in [m] \}$ in such a way that every subset $\mathcal{F} \subset \mathcal{H}$ of size at most $k$ has a \emph{transversal}, i.e., it is possible for each $F \in \mathcal{F}$ to pick a distinct point $p_F \in F$. We show that such a set $\mathcal{H}$ exists by a probabilistic argument based on Hall's marriage theorem.

\begin{restatable}{lemma}{lemhall} \label{lem-hall}
    We can choose $m = \Theta(\sqrt[3]{kn^2})$ and $\ell = \Theta(\sqrt[3]{k^2n})$ such that there exists a set $\mathcal{H}$ such that (a) $|\mathcal{H}| = m$, (b) $H_i \in \binom{[\ell]}{3}$ for each $H_i \in \mathcal{H}$, and (c) any subset $\mathcal{F} \subset \mathcal{H}$ of size at most $k$ admits a transversal. \label{lem:conjunctivegridcompression2}
\end{restatable}

This completes the description of the grid compression encoding; both correctness and a detailed count of the clauses are proved in the appendix, but in summary, constraint~\eqref{eq:varimpliesitscolinM} uses $n$ clauses, constraint~\eqref{eq:colcopyconjunctive} uses $3n$ clauses, and the remaining constraints all use $O(\sqrt[3]{kn^2})$ clauses.
This establishes \Cref{thm:conjunctive-compression}.

\subsection{Disjunctive grid compression}
Now we set $m = \Theta\lrp{\sqrt{nk\log_kn}}$ and $\ell = \Theta\lrp{k^2\log^2_kn}$ for
suitable choices of constants, and we will ensure that the $H_i$ are distinct and each of size
$\Theta(k \log_k n)$ for a suitable choice of constant.
In this encoding, the disjunctive switching manifests as the following set of clauses enforcing that each true variable in $M_{:,j}$ is copied to $L_{:,p}$ for \emph{some} column $p \in H_j$:
\begin{equation}
    \bigwedge_{(i, j) \in [\ceil{n / m}] \times [m]} \Big(\overline{M_{i, j}} \lor \bigvee_{p \in H_j} L_{i, p}\Big). \label{eq:colcopydisjunctive}
\end{equation}
This constraint, however, could allow true variables in different columns of $M$ to be copied to the same variable in $L$, and thus we need some mechanism to prevent this. To that end, we introduce auxiliary variables $\ov_p$ for each $p \in [\ell]$, which intuitively represent whether column $p$ of $L$ is \emph{overloaded}, meaning that there are multiple occupied columns of $M$ that could be copied to $p$. The semantics for $\ov_p$ are given by the following constraints:
\begin{equation}
    \bigwedge_{p \in [\ell]} \ov_p \lor \amo(\{c_j \mid j \in [m] \text{ such that } p \in H_j\}).\label{eq:getoverloadedcols}
\end{equation}

Finally, to make the disjunctive switching work, we enforce that if we copy into some $L_{i, p}$, then column $p$ must not be overloaded:
\begin{equation}
    \bigwedge_{i \in [\ceil{n/m}]} \bigwedge_{p \in [\ell]} \left(\overline{L_{i, p}} \lor \overline{\ov_p}\right). \label{eq:removeoverloadedcols}
\end{equation}

These are all of the constraints for the disjunctive grid compression encoding, and it remains to specify 
$\mathcal{H} := \{H_j \mid j \in [m]\}$. 
For the encoding to be correct, we require the following: if we pick any $k$ sets $H_{j_1}, \ldots, H_{j_k}  \in \mathcal{H}$, then $H_{j_1} \not\subseteq H_{j_2} \cup \dots \cup H_{j_k}$. This so-called \emph{$(k-1)$-cover-free property} has been studied by numerous authors \cite{kautz-singleton,cover-free, Idalino2023}. We prove the existence of a $(k-1)$-cover-free family $\mathcal{H}$ that satisfies our requirements using an elegant algebraic construction from \cite{cover-free}:
\begin{restatable}{lemma}{lemreedsolomon} \label{lem:disjointgridcompression3}
    There exists $q = \Theta(k \log_k n)$ such that there exists a $(k-1)$-cover-free set $\mathcal{H}$ satisfying $H_i \in \binom{[\ell]}{q}$ for each $H_i \in \mathcal{H}$ and $|\mathcal{H}| \ge m = \Theta\lrp{\sqrt{nk\log_kn}}$.
\end{restatable}

With $\mathcal{H}$ chosen according to the lemma, the encoding is complete. In contrast to the probabilistic proof of \Cref{lem-hall}, the construction in \Cref{lem:disjointgridcompression3} is deterministic and efficient.

\section{Concluding Remarks and Empirical Evaluation}\label{sec:discussion}

  

We solved a fundamental problem in the theory of CNF encodings by showing that the minimum number of clauses in an encoding of $\amk(x_1,\dots,x_n)$ is $2n + \widetilde{\Theta}(\sqrt{n})$ for each fixed $k$. We also tightened the upper bound on the minimum number of clauses in an encoding of $\amo$, refuting a conjecture of Chen~\cite{Chen2010ANS} and resolving a long-standing open problem in circuit complexity~\cite[p.~158]{bloniarz1979}. En route to these results, we introduced (a) a graph-theoretic framework for constructing and analyzing $\amo$ encodings and (b) disjunctive switching, a general technique for compactly encoding switch statements into CNF, and (c) grid compression, a framework for constructing CNF encodings of boolean functions supported on inputs of small Hamming weight.

Our constructions are structurally very different from other encodings in the literature; for instance, the multipartite encoding for $\amo$ is notable for using clauses of width 3, despite the fact that $\amo$ is a 2-CNF function. Ku{\v c}era, Savick{\'{y}}, and Vorel~\cite{lower-bound} asked whether the smallest propagation-complete encoding of an antitone 2-CNF function is always 2-CNF. While an affirmative answer has some \emph{prima facie} plausibility, our construction concretely demonstrates how wide clauses can be useful even for $\amo$, the simplest antitone 2-CNF function. In fact, we conjecture that Chen's product encoding is essentially optimal for 2-CNF encodings, in the sense that every 2-CNF encoding for $\amo$ has at least $2n + 4 \sqrt{n} - o(\sqrt{n})$ clauses.

Our encodings from \Cref{thm:amk-disjunctive-product,thm:disjunctive-compression} are also unique for employing disjunctive switching, which has not been used in any previous CNF encodings to the best of our knowledge. Notably, they allow us to surpass the $3n$ lower bound for monotone circuits~\cite{sergeev}.
Thus, although previous researchers have noted the close connection between CNF encodings for $\amk$ and circuits for threshold functions (see, e.g., \cite{lower-bound,sergeev,cnf-symmetric}), our results demonstrate the importance of viewing CNF encodings as constituting an important model of computation in their own right. We expect disjunctive switching, and the innovative use of wide clauses more generally, to be important for harnessing the full power of CNF encodings in other contexts as well.


The above considerations demonstrate how rich the theory of CNF encodings is, even for very simple boolean functions. Given the importance of CNF encodings to SAT solving, and the fact that ``not much is known about CNF encodings from a theoretical point of view''~\cite{cnf-symmetric}, we expect the further development of this theory to be of great practical significance.

\subparagraph*{Empirical results.}
Even though propagation completeness is often described as an essential requirement for practical encodings of cardinality constraints~\cite[p. 95]{karpinski2019cnfencodingscardinalityconstraints}, exploratory experiments with our encodings show they can be competitive in practice---especially the disjunctive grid compression (DGC) encoding. \Cref{fig:enc-comparison} depicts results on a family of instances in which one would expect propagation completeness to play a role; $3$ random disjoint subsets of 10 variables (out of $n$) are chosen, and we add clauses enforcing that at least one variable in each subset is true. Then, a global $\amt$ constraint makes the instance unsatisfiable. We compared against all cardinality encodings present in PySAT~\cite{imms-sat18}, of which the \emph{sequential counter}~\cite{sinz} performed best. The number of clauses does \emph{not} reliably predict which encoding performed best, and yet this proxy led us to develop encodings that perform well in practice.
Appendix~\ref{sec:empirical} includes a more detailed experimental evaluation, but in a nutshell, we believe our results challenge the folklore claim of propagation completeness being absolutely necessary.

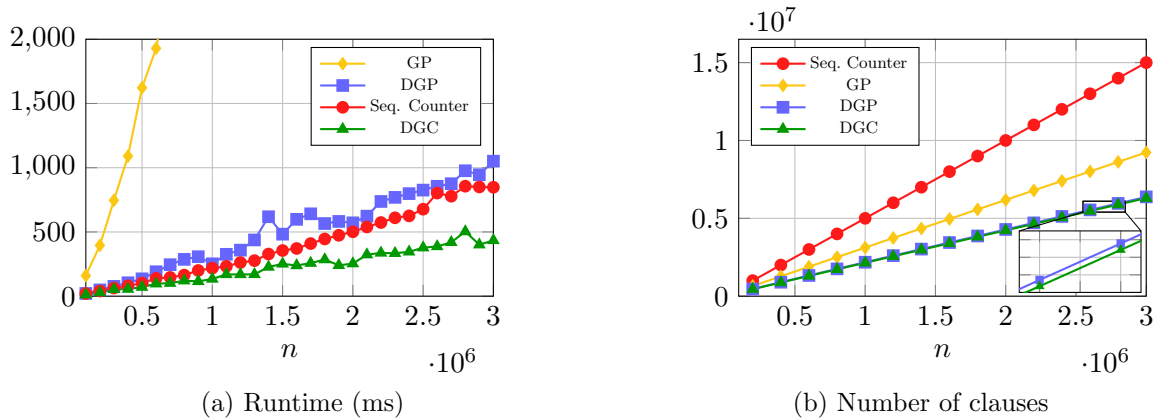
\begin{figure}
    \begin{subfigure}{0.49\linewidth}
        \begin{tikzpicture}
\begin{axis}[
  title={},
  xlabel={$n$},
  grid=major,
  legend pos=north east,
  xmin = 100000,
  xmax = 3000001,
  ymin = 0,
  ymax = 2000,
  height=5.0cm,
  width=7cm,
  legend style={
    font=\footnotesize,
    nodes={scale=0.7, transform shape}
}
]

\addplot+[mark=diamond*, lipicsYellow, mark options={fill=lipicsYellow}, line width=0.8pt] coordinates {(100000, 159.8) (200000, 395.6) (300000, 746.4) (400000, 1090.4) (500000, 1622.4) (600000, 1929.4) (700000, 2504.2) (800000, 2550.4) (900000, 2691) (1000000, 3170.6) (1100000, 3938.8) (1200000, 4185.4) (1300000, 3815.8) (1400000, 5657.8) (1500000, 5163.4) (1600000, 5879.8) (1700000, 6570) (1800000, 6208.6) (1900000, 6756.4) (2000000, 6502.6) (2100000, 8226.8) (2200000, 7014.8) (2300000, 9410) (2400000, 8657.6) (2500000, 9856) (2600000, 9959.6) (2700000, 10839) (2800000, 12628.4) (2900000, 12601) (3000000, 10012.6)};
\addlegendentry{GP}

\addplot+[mark=square*, blue!60, mark options={fill=blue!60}, line width=0.8pt] coordinates {(100000, 21.4) (200000, 49) (300000, 76.8) (400000, 104.6) (500000, 135.2) (600000, 190.4) (700000, 245.2) (800000, 286.6) (900000, 307.4) (1000000, 252.2) (1100000, 326) (1200000, 358.4) (1300000, 437.8) (1400000, 618.4) (1500000, 482.8) (1600000, 598.2) (1700000, 641.2) (1800000, 565.6) (1900000, 582) (2000000, 571) (2100000, 623.2) (2200000, 737.4) (2300000, 769.2) (2400000, 798.4) (2500000, 825.8) (2600000, 855.6) (2700000, 875.6) (2800000, 976.4) (2900000, 945.4) (3000000, 1051)};
\addlegendentry{DGP}
\addplot+[mark=*, red!90!white, mark options={fill=red!90!white}, line width=0.8pt] coordinates {(100000, 20.6) (200000, 40.4) (300000, 63.4) (400000, 80.8) (500000, 101.2) (600000, 139) (700000, 146.2) (800000, 164.2) (900000, 200.2) (1000000, 218.4) (1100000, 232.2) (1200000, 261.8) (1300000, 276.8) (1400000, 328) (1500000, 355.4) (1600000, 371.4) (1700000, 409.6) (1800000, 445.4) (1900000, 474) (2000000, 501.4) (2100000, 538.8) (2200000, 573.4) (2300000, 610.2) (2400000, 625.2) (2500000, 677.2) (2600000, 802.6) (2700000, 778.4) (2800000, 856.2) (2900000, 850.2) (3000000, 848)};
\addlegendentry{Seq. Counter}
\addplot+[mark=triangle*, green!60!black, mark options={fill=green!60!black}, line width=0.8pt] coordinates {(100000, 11.8) (200000, 28) (300000, 53.4) (400000, 55.4) (500000, 70.6) (600000, 97) (700000, 101.6) (800000, 120.6) (900000, 115.4) (1000000, 134.2) (1100000, 169.6) (1200000, 169.4) (1300000, 170) (1400000, 227.2) (1500000, 251.6) (1600000, 239) (1700000, 257.4) (1800000, 285.2) (1900000, 240.4) (2000000, 253.6) (2100000, 323.4) (2200000, 336.8) (2300000, 334.2) (2400000, 345.6) (2500000, 377.6) (2600000, 387.2) (2700000, 419.4) (2800000, 503.2) (2900000, 400.8) (3000000, 434.4)};
\addlegendentry{DGC}
\end{axis}
\end{tikzpicture}
\caption{Runtime (ms)}
    \end{subfigure}
    \begin{subfigure}{0.49\linewidth}
    \centering
        \begin{tikzpicture}
  \begin{axis}[
    name=main,
    xlabel={$n$},
    grid=major,
    legend pos=north west,
    xmin = 100000, xmax = 3000001,
    ymin = 0,
    height=5.0cm, width=7cm,
    legend style={font=\footnotesize, nodes={scale=0.7, transform shape}},
  ]

    \addplot+[mark=*, red!90!white, mark options={fill=red!90!white}, line width=0.8pt]
    coordinates {(200000, 999993) (400000, 1999993) (600000, 2999993) (800000, 3999993)
      (1000000, 4999993) (1200000, 5999993) (1400000, 6999993) (1600000, 7999993)
      (1800000, 8999993) (2000000, 9999993) (2200000, 10999993) (2400000, 11999993)
      (2600000, 12999993) (2800000, 13999993) (3000000, 14999993)};
  \addlegendentry{Seq.\ Counter}

\addplot+[mark=diamond*, lipicsYellow, mark options={fill=lipicsYellow}, line width=0.8pt] coordinates {(200000, 654117) (400000, 1273908) (600000, 1892916) (800000, 2506839) (1000000, 3120159) (1200000, 3737979) (1400000, 4349103) (1600000, 4959408) (1800000, 5571486) (2000000, 6181791) (2200000, 6793356) (2400000, 7401942) (2600000, 8011734) (2800000, 8622291) (3000000, 9232587)};
\addlegendentry{GP}

  \addplot+[mark=square*, blue!60, mark options={fill=blue!60}, line width=0.8pt]
    coordinates {(200000, 462163) (400000, 897953) (600000, 1329347) (800000, 1754915)
      (1000000, 2179177) (1200000, 2605203) (1400000, 3024873) (1600000, 3445443)
      (1800000, 3866913) (2000000, 4284737) (2200000, 4707827) (2400000, 5122113)
      (2600000, 5541665) (2800000, 5956707) (3000000, 6377267)};
  \addlegendentry{DGP}

  \addplot+[mark=triangle*, green!60!black, mark options={fill=green!60!black}, line width=0.8pt]
    coordinates {(200000, 448996) (400000, 877578) (600000, 1301906) (800000, 1723022)
      (1000000, 2143170) (1200000, 2561360) (1400000, 2979142) (1600000, 3395660)
      (1800000, 3811884) (2000000, 4227056) (2200000, 4641736) (2400000, 5056372)
      (2600000, 5470200) (2800000, 5883808) (3000000, 6297534)};
  \addlegendentry{DGC}

  \draw[black]
    (axis cs:2550000, 5400000) rectangle (axis cs:2850000, 6100000);

\coordinate (zSW) at (axis cs:2550000,5400000);
\coordinate (zNE) at (axis cs:2850000,6100000);
\coordinate (zSE) at (axis cs:2850000,5400000);
\draw[black] (zSW) rectangle (zNE);

  \end{axis}

  \begin{axis}[
    name=zoom,
    at={(main.south east)}, anchor=south east,
    xshift=-0.7mm, yshift=0.5mm, 
    width=3.2cm, height=2.4cm,
    xmin=2550000, xmax=2850000,
    ymin=5400000, ymax=6100000,
    grid=major,
    xticklabel=\empty,
    yticklabel=\empty,
     scaled x ticks=false,
    scaled y ticks=false,
     xlabel={}, ylabel={},
     axis background/.style={fill=white},
  ]

  \addplot+[mark=square*, blue!60, mark options={fill=blue!60}, line width=0.8pt,
            mark size=1.2pt] coordinates {(200000, 462163) (400000, 897953) (600000, 1329347)
      (800000, 1754915) (1000000, 2179177) (1200000, 2605203) (1400000, 3024873)
      (1600000, 3445443) (1800000, 3866913) (2000000, 4284737) (2200000, 4707827)
      (2400000, 5122113) (2600000, 5541665) (2800000, 5956707) (3000000, 6377267)};

  \addplot+[mark=triangle*, green!60!black, mark options={fill=green!60!black}, line width=0.8pt,
            mark size=1.4pt] coordinates {(200000, 448996) (400000, 877578) (600000, 1301906)
      (800000, 1723022) (1000000, 2143170) (1200000, 2561360) (1400000, 2979142)
      (1600000, 3395660) (1800000, 3811884) (2000000, 4227056) (2200000, 4641736)
      (2400000, 5056372) (2600000, 5470200) (2800000, 5883808) (3000000, 6297534)};


  \end{axis}

\draw[black, thin] (zSW) -- (zoom.north west);
\draw[black, thin] (zSE) -- (zoom.north east);
\end{tikzpicture}
\caption{Number of clauses}
    \end{subfigure}
    \caption{Comparison of encodings for $\amt$ (UNSAT). We abbreviate the generalized product encoding as GP and its disjunctive variant as DGP.}
    \label{fig:enc-comparison}
\end{figure}

\section*{Acknowledgments}
We are grateful to Petr Savick\'y for finding an error in an earlier version of this paper. We also thank Marijn Heule for his suggestions on improving the presentation of this paper.

\section*{Funding}
Krapivin was supported by the NSF grant CNS2504471, Packard Foundation grant 2020-71730, and the Jeanne B. and Richard F. Berdik ARCS Pittsburgh Endowed Scholar Award.
    Przybocki and Subercaseaux were supported by NSF grant DMS-2434625. Przybocki was additionally supported by the NSF Graduate Research Fellowship Program under Grant
No. DGE-2140739. 

\bibliographystyle{abbrvurl}
\bibliography{bib}

\appendix

\section{\texorpdfstring{Proofs for \Cref{sec-multipartite}}{Proofs for the multipartite encoding}}
\subsection{\texorpdfstring{Proof of \Cref{lem-amo-indicator}}{Proof of the indicator lemma}}

\lemamoindicator*
\begin{proof}
    Rename the input variables $x_1,\dots,x_n$ to be of the form $x_{i,j}$ with $i,j \in [p]$, where $p = \ceil{\sqrt{n}}$. Let $E$ be the set of ordered pairs $(i,j)$ to which some variable is assigned. Then, we can use the following extension of the product encoding:
    \begin{multline*}
    \amo'(\{ x_{i, j} \mid (i, j) \in E\};z) := \left(\bigwedge_{(i, j) \in E} (\overline{x_{i, j}} \lor r_i) \land  (\overline{x_{i, j}} \lor c_j) \right) \\ \land \textsf{PE}(r_1, \ldots, r_p) \land \textsf{PE}(c_1, \ldots, c_p) \land 
    \bigwedge_{i \in [p]} (\overline{r_i} \lor z).
    \end{multline*}
    It is not hard to see that the encoding uses  $2n + O(\sqrt{n})$ clauses and $O(\sqrt{n})$ auxiliary variables, and the justification of the correctness and propagation completeness of this encoding is very similar to that of the product encoding.
\end{proof}

\subsection{\texorpdfstring{Proof of \Cref{thm:circuit}}{Proof of the threshold-2 circuit theorem}}

Now we describe how the same construction yields a circuit for $T_2(x_1,\dots,x_n)$. Let $S_2 \colon \{\bot, \top\}^n \to \{\bot, \top\}^{2}$ be the boolean operator defined by 
\[ 
S_2(x_1, \dots, x_n) := \lrp{x_1 \lor \dots \lor x_n,\; T_2(x_1, \dots, x_n)}.
\]
We make use of two standard facts about circuits for threshold functions. First, as an analogue of \Cref{lem-amo-indicator}, $S_2$ has a circuit of size $2n + O(\sqrt{n})$~\cite{sergeev}. Second, as an analogue of the generalized product encoding, $T_3$ has a circuit of size $3n + O(n^{2/3})$~\cite{dunne1984,wegener1987}, where $T_3$ is the negation of $\amt$. Note that circuits for $T_3$ cannot be single level, and thus neither will be our circuit for $T_2$, which internally uses a $T_3$ circuit.

\thmcircuit*
\begin{proof}
    Let $G$ be the complete $p$-partite graph with $q$ vertices within each part, where $p = \ceil{\sqrt[6]{n}}+1$ and $q = \ceil{\sqrt{2} \cdot \sqrt[3]{n}}$, so $|E(G)| \ge n$. Let $P_1,\dots,P_p$ be the parts of $G$. Assign the variables $x_1,\dots,x_n$ to distinct edges of $G$, renaming the variables so that $x_e$ is the variable assigned to the edge $e$. Let $E$ be the set of edges of $G$ to which some variable is assigned. Discard any vertices of $G$ not incident to an edge from $E$.

    For each $i \in V(G)$, let $y_i = \bigvee_{\substack{e \in E \\ i \in e}} x_e$. Then, let $(z_k,w_k) = S_2(\{y_i \mid i \in P_k\})$ for each $k \in [p]$. Then, our circuit for $T_2$ is as follows:
    \[
        T_2(\{x_e \mid e \in E\}) := \bigvee_{k \in [p]} w_k \lor T_3(z_1,\dots,z_p).
    \]

    The justification for the correctness of the circuit is similar to the proof of \Cref{thm:am1}. The number of gates required to compute all of the $y_i$ variables is at most $2n - \sqrt{2n}$. The number of gates required to compute all of the $(z_k,w_k)$ variables is at most $p \cdot (2q + O(\sqrt{q})) = 2 \sqrt{2n} + O(\sqrt[3]{n})$. In total, the gate complexity of the circuit is $2n + \sqrt{2n} + O(\sqrt[3]{n})$, as desired.
\end{proof}

\section{\texorpdfstring{A Lower Bound for $\amo$}{A Lower Bound for AMO}}
\begin{definition}
    Let $\vec{x} = (x_1,\dots,x_n)$. We say that $\varphi(\vec{x},\vec{y})$ is a \emph{generalized $P$-encoding on $n$ input variables} if it satisfies the following two conditions:
    \begin{enumerate}[label=(\alph*)]
        \item $\varphi \land x_i$ is satisfiable for each $i \in [n]$, and
        \item $\varphi \models \overline{x_i} \lor \overline{x_j}$ for each $i,j \in [n]$ with $i \neq j$.
    \end{enumerate}
\end{definition}
Ku{\v c}era, Savick{\'{y}}, and Vorel~\cite{lower-bound} defined $P$-encodings almost equivalently, but with condition (b) using $\vdash_1$ (i.e., derivable by unit propagation) instead of our semantic implication $\models$. Our proof will leverage a few auxiliary lemmas of their work, whose proofs apply directly in our context by simply changing each $\vdash_1$ for a $\models$. Nonetheless, in some cases (e.g., our~\Cref{lemma:no-single-clause}) we had to use a different proof since theirs relied essentially on unit propagation.

It is easy to see that $\varphi(\vec{x},\vec{y})$ is a generalized $P$-encoding if and only if it is an encoding of $\amo(x_1,\dots,x_n)$ or $\eo(x_1,\dots,x_n)$, where $\eo(x_1,\dots,x_n)$ is the boolean function that evaluates to true if and only if exactly one of its inputs is true.

Throughout the proof, for a CNF formula $\varphi$, and an input variable $x_i$, we will use notation $Q_{\varphi, i} := \{ C \in \varphi : \overline{x_i} \in C\}$ for the set of clauses containing $x_i$ negatively.

\begin{definition}
    Let $\varphi(\vec{x},\vec{y})$ be a generalized $P$-encoding on $n$ input variables. We say that $\varphi$ is in \emph{regular form} if, for each $i \in [n]$, we have the following:
    \begin{enumerate}[label=(\alph*)]
        \item $|Q_{\varphi,i}| = 2$,
        \item $x_i$ is the only input variable in the clauses in $Q_{\varphi,i}$, and
        \item each clause in $Q_{\varphi,i}$ is binary.
    \end{enumerate}
\end{definition}

We say that a generalized $P$-encoding $\varphi$ on $n$ variables is \emph{prime} if any CNF $\varphi'$ obtained by removing any literal from any of the clauses of $\varphi$, is not a generalized $P$-encoding on $n$ variables.

\begin{lemma} \label{lem-regular}
    If $\varphi(\vec{x},\vec{y})$ is a prime generalized $P$-encoding on $n$ input variables minimizing $|\varphi|$ for some $n \ge 4$, then either there is some generalized $P$-encoding $\varphi'$ on $n$ input variables in regular form and such that $|\varphi'| = |\varphi|$, or there is a generalized $P$-encoding $\varphi'$ on $n-1$ input variables such that $|\varphi| \ge |\varphi'| + 3$.
\end{lemma}

\subsection{\texorpdfstring{Proof of \Cref{lem-regular}}{Proof of the regular-form lemma}}

Recall that in a generalized $P$-encoding $\varphi$, the formula  $\varphi \land x_i$ is satisfiable for every $i$. Let us use notation $\tau_i$ for an arbitrary satisfying assignment of $\varphi \land x_i$.

\begin{lemma}\label{lemma:no-single-clause}
        Let $\varphi(\vec{x},\vec{y})$ be a prime generalized $P$-encoding on $n$ input variables minimizing $|\varphi|$ for some $n \ge 3$. Then,  $|Q_{\varphi,i}| > 1$ for every $i \in [n]$.
\end{lemma}
\begin{proof}
Assume, expecting a contradiction, that for some $i$ there is only a single clause $C' := (\overline{x_i} \lor C) \in \varphi$ containing $\overline{x_i}$. We now show by cases that $C$ cannot contain input variables:
    \begin{itemize}
        \item If $x_j \in C$ for some $j \neq i$, then  $\tau'(x) := \begin{cases}
            \top & \text{if } x = x_i\\
            \tau_j(x) & \text{otherwise}
        \end{cases}$ satisfies $\varphi$ and yet $\tau' \not \models \overline{x_i} \lor \overline{x_j}$, a contradiction. If $x_i \in C$, then $C$ is a tautology, so removing it contradicts minimality of $|\varphi|$.
        \item If $\overline{x_j} \in C$ for some $j \neq i$, then let $k \in [n] \setminus \{i, j\}$, and note that $\tau'(x) := \begin{cases}
            \top & \text{if } x = x_i\\
            \tau_k(x) & \text{otherwise}
        \end{cases}$ satisfies $\varphi$ (since $\tau_k \models \overline{x_j}$) and yet $\tau' \not \models \overline{x_i} \lor \overline{x_k}$, a contradiction.
    \end{itemize}
    Since $\varphi \land x_i$ is satisfiable, $C$ cannot be empty, and thus there is some auxiliary literal $\ell \in C$. Now, let $\varphi' := \varphi \setminus \{C'\}$, and observe that $\varphi' \models \overline{\ell} \lor \overline{x_j}$ for every $j \neq i$. Indeed, if $\varphi' \land \ell \land x_j$ were satisfied by some assignment $\tau$, 
    then $\tau'(x)  := \begin{cases}
            \top & \text{if } x = x_i\\
            \tau(x) & \text{otherwise}
        \end{cases}$
    would satisfy both $\varphi'$ and $C'$, and thus $\varphi$, contradicting that $\varphi \not \models \overline{x_i} \lor \overline{x_j}$. Moreover, $\varphi' \land \ell$ must be satisfiable, since if it were not, and $\varphi' \models \overline{\ell}$, then removing $\ell$ from $C$ in $\varphi$ would still yield a generalized $P$-encoding, contradicting $\varphi$ being prime.  Note that, for distinct indices $j, k \in [n] \setminus \{i\}$, we have $\varphi' \models \overline{x_j} \lor \overline{x_k}$, since if there were some satisfying assignment for $\varphi' \land x_j \land x_k$, extending it to set $x_i = \bot$ would yield a satisfying assignment for $\varphi \land x_j \land x_k$, a contradiction. Thus, letting $\vec{x}' := (x_1, \dots, x_{i-1}, \ell, x_{i+1}, \dots, x_n)$, we have that $\varphi'$ is a generalized $P$-encoding on $n$ input variables, and since $|\varphi'| = |\varphi| - 1$, we have reached a contradiction.
\end{proof}

\begin{lemma} \label{lem-regular-1}
    Let $\varphi(\vec{x},\vec{y})$ be a prime generalized $P$-encoding on $n$ input variables minimizing $|\varphi|$ for some $n \ge 3$, and suppose that $|Q_{\varphi,i}| \neq 2$ for some $i \in [n]$. Then, there is a generalized $P$-encoding $\varphi'$ on $n-1$ input variables such that $|\varphi| \ge |\varphi'| + 3$.
\end{lemma}
\begin{proof}
    The proof is obtained by replacing each $\vdash_1$ with $\models$ in that of~\cite[Proposition~4.5]{lower-bound}.
\end{proof}

\begin{lemma} \label{lem-regular-2}
    Let $\varphi(\vec{x},\vec{y})$ be a prime generalized $P$-encoding on $n$ input variables minimizing $|\varphi|$ for some $n \ge 4$, and suppose that $|Q_{\varphi,i}| = 2$ for all $i \in [n]$, but there is some $i \in [n]$ and $C \in Q_{\varphi,i}$ such that $C$ contains an input variable other than $x_i$. Then, there is a generalized $P$-encoding $\varphi'$ on $n-1$ input variables such that $|\varphi| \ge |\varphi'| + 3$.
\end{lemma}
\begin{proof}
        The proof is obtained by replacing each $\vdash_1$ with $\models$ in that of~\cite[Proposition~4.7]{lower-bound}.
\end{proof}

\begin{proof}[Proof of \Cref{lem-regular}]
    Let $\varphi(\vec{x},\vec{y})$ be a prime generalized $P$-encoding on $n$ input variables minimizing $|\varphi|$ for some $n \ge 4$. By \Cref{lem-regular-1,lem-regular-2}, we may assume that for all $i \in [n]$, we have $|Q_{\varphi,i}| = 2$ and $x_i$ is the only input variable in the clauses in $Q_{\varphi,i}$. If every clause in $\bigcup_{i \in [n]} Q_{\varphi,i}$ is binary, then we are done. Otherwise, it suffices to construct a generalized $P$-encoding $\varphi'$ on $n$ input variables such that:
    \begin{enumerate}[label=(\alph*)]
        \item $\varphi'$ has one fewer non-binary clause in $\bigcup_{i \in [n]} Q_{\varphi',i}$ compared to $\varphi$,
        \item $|\varphi'| = |\varphi|$, and
        \item for all $i \in [n]$, we have $|Q_{\varphi',i}| = 2$ and $x_i$ is the only input variable in the clauses in $Q_{\varphi',i}$.
    \end{enumerate}
    Indeed, repeating the argument with respect to $\varphi'$ must eventually lead to a generalized $P$-encoding in regular form. So let us see how to construct such a $\varphi'$.
    
    Suppose $C \in Q_{\varphi,i}$ is not binary. By~\cite[Lemma~4.3]{lower-bound}, we have $C = (\overline{x_i} \lor \ell_1 \lor \dots \lor \ell_k)$ for some literals $\ell_1, \dots, \ell_k \in \mathrm{lit}(\vec{y})$. Therefore, $\tau_i \models \ell_j$ for some $j \in [k]$. Now, let
    \[
        \varphi' = (\varphi \setminus \{C\}) \cup \{(\overline{x_i} \lor \ell_j)\}.
    \]
    It remains to show that $\varphi'$ is a generalized $P$-encoding. By our choice of $\ell_j$, we have $\tau_i \models \varphi' \land x_i$, and for any $j \neq i$, the formula $\varphi' \land x_{j}$ is satisfiable, since any satisfying assignment for $\varphi \land x_{i'}$ is also one for $\varphi' \land x_{i'}$. Since $\varphi' \models \varphi$, we also have $\varphi' \models \overline{x_{i}} \lor \overline{x_{j}}$ for all $i,j \in [n]$ with $i \neq j$. Thus, $\varphi'$ is a generalized $P$-encoding satisfying the above conditions.
\end{proof}

\subsection{\texorpdfstring{Bounding the size of generalized $P$-encodings in regular form}{Bounding the size of generalized P-encodings in regular form}}

\begin{lemma} \label{lem-weaker-lb}
    If $\varphi(\vec{x},\vec{y})$ is a generalized $P$-encoding on $n$ input variables for some $n \ge 3$, then $|\varphi| \ge \min(2n,3n-6)$.
\end{lemma}
\begin{proof}
    We may assume that $\varphi$ is prime, since otherwise we can simply remove redundant literals until we have a prime encoding. The proof is by induction on $n$. For the base case $n = 3$, we can verify directly that $|\varphi| \ge 3$, so assume that $n \ge 4$.  If $\varphi$ is in regular form, then $|\varphi| \geq \left|\bigcup_{i=1}^n Q_{i, \varphi}\right| = 2n$, and we are done.
   Otherwise, by~\Cref{lem-regular}, there is a generalized $P$-encoding $\varphi'$ on $n-1$ input variables such that $|\varphi| \ge |\varphi'| + 3$, and by the inductive hypothesis $|\varphi'| + 3 \geq \min(2n-2, 3n-9) + 3 \geq \min(2n, 3n-6)$, completing the proof.
\end{proof}

For a $P$-encoding $\varphi$, and input variable index $i$,  let $L_{\varphi, i} := \{ e : (\overline{x_i} \lor e) \in \varphi \}$.

\begin{lemma} \label{lem-distinct-edges}
    If $\varphi(\vec{x},\vec{y})$ is a generalized $P$-encoding on $n$ input variables in regular form and $i,j \in [n]$ with $i \neq j$, then $L_{\varphi,i} \neq L_{\varphi,j}$.
\end{lemma}
\begin{proof}
      The proof is obtained by replacing each $\vdash_1$ with $\models$ in that of~\cite[Lemma~4.8]{lower-bound}.
\end{proof}

We use notation $\mathrm{var}(\ell)$ to denote the variable corresponding to literal $\ell$, meaning that $\mathrm{var}(x) = \mathrm{var}(\overline{x})= x$ for each variable $x$.

\begin{lemma} \label{lem-pairwise-distinct}
    Let $\varphi(\vec{x},\vec{y})$ be a generalized $P$-encoding on $n$ input variables in regular form, and suppose that there are distinct $i,j,k \in [n]$ with $L_{\varphi,i} = \{\ell,\ell_1\}$, $L_{\varphi,j} = \{\ell,\ell_2\}$, and $L_{\varphi,k} = \{\ell,\ell_3\}$ for some $\ell,\ell_1,\ell_2,\ell_3 \in \mathrm{lit}(\vec{y})$. Then, the variables $\mathrm{var}(\ell_1)$, $\mathrm{var}(\ell_2)$, and $\mathrm{var}(\ell_3)$ are pairwise distinct.
\end{lemma}
\begin{proof}
      The proof is obtained by replacing each $\vdash_1$ with $\models$ in that of~\cite[Lemma~5.2]{lower-bound}.
\end{proof}

For each $i \in [n]$, we now let $\tau_i \colon \vec{x} \to \{\bot,\top\}$ be the assignment such that $\tau_i(x_i) = \top$ and $\tau_i(x_j) = \bot$ for $j \neq i$. For distinct $i,j \in [n]$, let $\tau_{i,j} \colon \vec{x} \to \{\bot,\top\}$ be the assignment such that $\tau_{i,j}(x_i) = \tau_{i,j}(x_j) = \top$ and $\tau_{i,j}(x_k) = \bot$ for $k \notin \{i,j\}$.

\begin{lemma} \label{lem-two-edge-conflict}
    Let $\varphi(\vec{x},\vec{y})$ be a generalized $P$-encoding on $n$ input variables in regular form, let
    \(
        \psi(\vec{x},\vec{y}) = \varphi \setminus \bigcup_{k \in [n]} Q_{\varphi,k},
    \)
    and let distinct $i,j \in [n]$ with $L_{\varphi,i} = \{\ell_{i,1},\ell_{i,2}\}$ and $L_{\varphi,j} = \{\ell_{j,1},\ell_{j,2}\}$. Then,
    \(
        \psi|_{\tau_i} \land \ell_{i,1} \land \ell_{i,2} \land \ell_{j,1} \land \ell_{j,2} \models \bot.
    \)
\end{lemma}
\begin{proof}
    Since $\varphi$ is a generalized $P$-encoding, we have $\varphi|_{\tau_{i,j}} \models \bot$. Moreover,
    \[
        \psi|_{\tau_{i,j}} \land \ell_{i,1} \land \ell_{i,2} \land \ell_{j,1} \land \ell_{j,2}
    \]
    is precisely $\varphi|_{\tau_{i,j}}$, because the clauses in $Q_{\varphi,i}$ and $Q_{\varphi,j}$ reduce to the displayed unit clauses, while every clause in $Q_{\varphi,k}$ for $k \notin \{i,j\}$ is satisfied by $\tau_{i,j}$. Thus,
    \[
        \psi|_{\tau_{i,j}} \land \ell_{i,1} \land \ell_{i,2} \land \ell_{j,1} \land \ell_{j,2} \models \bot.
    \]
    Finally, $\psi|_{\tau_{i,j}} \subseteq \psi|_{\tau_i}$, since $\psi$ contains no negative occurrences of input variables and any clause of $\psi$ that survives the restriction by $\tau_{i,j}$ also survives the restriction by $\tau_i$. Hence, the formula in the statement is unsatisfiable as well.
\end{proof}

\begin{lemma} \label{lem-switching-witness}
    Let $F$ be a CNF formula, let $\alpha$ be a satisfying assignment for $F$, and let $\alpha_1,\dots,\alpha_t$ be assignments that falsify $F$. Then, for each $i \in [t]$, we can choose a clause $C_i \in F$ and a literal $m_i$ such that $\alpha_i \models m_i$, $\alpha \models \overline{m_i}$, and $\overline{m_i} \in C_i$. Furthermore, if $C_i = C_j$ and $m_i \neq m_j$, then $\alpha_i \models m_j$ and $\alpha_j \models m_i$.
\end{lemma}
\begin{proof}
    For each $i \in [t]$, choose a clause $C_i \in F$ falsified by $\alpha_i$. Since $\alpha \models F$, some literal $r_i \in C_i$ is satisfied by $\alpha$. Since $\alpha_i$ falsifies $C_i$, we have $\alpha_i \not\models r_i$. Thus, with $m_i := \overline{r_i}$, we obtain $\alpha_i \models m_i$, $\alpha \models \overline{m_i}$, and $\overline{m_i} \in C_i$.

    Now suppose that $C_i = C_j$ and $m_i \neq m_j$. Since $\overline{m_j} \in C_j = C_i$ and $\alpha_i$ falsifies $C_i$, we have $\alpha_i \models m_j$. The proof that $\alpha_j \models m_i$ is symmetric.
\end{proof}

\begin{lemma} \label{lem-regular-lb}
    If $\varphi(\vec{x},\vec{y})$ is a generalized $P$-encoding on $n$ input variables in regular form for some $n \ge 8$, then $|\varphi| \ge 2n + \sqrt{n+1} - 2$.
\end{lemma}
\begin{proof}
    Let $\varphi(\vec{x},\vec{y})$ be a generalized $P$-encoding on $n$ input variables in regular form for some $n \ge 8$. Let
    \(
        \psi(\vec{x},\vec{y}) = \varphi \setminus \bigcup_{i \in [n]} Q_{\varphi,i}.
    \)
    Since $|\bigcup_{i \in [n]} Q_{\varphi,i}| = 2n$, our goal is to show that $|\psi| \ge \sqrt{n+1} - 2$.
    
    Let $G$ be the graph defined by $V(G) = \bigcup_{i \in [n]} L_{\varphi,i}$ and $E(G) = \{L_{\varphi,i} \mid i \in [n]\}$. By \Cref{lem-distinct-edges}, $|E(G)| = n$. Hence, if $\Delta$ denotes the maximum degree of $G$ and $\tau$ the minimum size of a vertex cover of $G$, we have  $\Delta \cdot \tau \geq E(G) = n$, and thus either $\Delta > \sqrt{n+1} - 1$ or $\tau \ge \sqrt{n+1} + 1$.

    Suppose first that $\Delta > \sqrt{n+1} - 1$. Let $\ell \in V(G)$ be such that $\deg(\ell) = \Delta$. Without loss of generality, $\{L_{\varphi,i} \mid i \in [\Delta]\}$ are the edges incident to $\ell$. Write $L_{\varphi,i} = \{\ell,\ell_i\}$ for each $i \in [\Delta]$. Since $\Delta \ge 3$, \Cref{lem-pairwise-distinct} implies that $\mathrm{var}(\ell_i) \neq \mathrm{var}(\ell_j)$ for all $i,j \in [\Delta]$ with $i \neq j$; hence, without loss of generality, $\ell_i$ is an auxiliary variable (rather than the negation thereof) for each $i \in [\Delta]$. Since $\varphi|_{\tau_1}$ is satisfiable, so is $\psi|_{\tau_1} \land \ell \land \ell_1$. Let $\alpha : \vec{x} \cup \vec{y} \to \{\bot,\top\}$ be a satisfying assignment for this formula. For each $i \in [2,\Delta]$, \Cref{lem-two-edge-conflict} yields
    \(
        \psi|_{\tau_1} \land \ell \land \ell_1 \land \ell_i \models \bot,
    \)
    and therefore $\alpha \models \overline{\ell_i}$. For each $i \in [2,\Delta]$, let $\alpha_i$ be the assignment obtained from $\alpha$ by setting $\ell_i$ to $\top$. Since $\alpha_i \models \ell \land \ell_1 \land \ell_i$, the displayed entailment implies that $\alpha_i$ falsifies $\psi|_{\tau_1}$. Applying \Cref{lem-switching-witness} to the formula $\psi|_{\tau_1}$, the satisfying assignment $\alpha$, and the assignments $\alpha_2,\dots,\alpha_\Delta$, we obtain clauses $C_i \in \psi|_{\tau_1}$ and literals $m_i$ such that $\alpha_i \models m_i$, $\alpha \models \overline{m_i}$, and $\overline{m_i} \in C_i$. Since $\alpha_i$ differs from $\alpha$ only on the variable $\mathrm{var}(\ell_i)$ and $\alpha \models \overline{\ell_i}$, we must have $m_i = \ell_i$. Hence, if $i,j \in [2,\Delta]$ are distinct and $C_i = C_j$, then \Cref{lem-switching-witness} yields $\alpha_i \models \ell_j$, contradicting $\mathrm{var}(\ell_i) \neq \mathrm{var}(\ell_j)$. Thus, the clauses $C_2,\dots,C_\Delta$ are pairwise distinct, so $|\psi| \ge |\psi|_{\tau_1}| \ge \Delta - 1 \ge \sqrt{n+1} - 2$, as desired.

    Next, suppose that $\tau \ge \sqrt{n+1} + 1$. Write $L_{\varphi,i} = \{\ell_{i,1},\ell_{i,2}\}$ for each $i \in [n]$, and let
    \(
        F = \psi|_{\tau_1} \land \ell_{1,1} \land \ell_{1,2}.
    \)
    Since $\varphi|_{\tau_1}$ is satisfiable, so is $F$. Moreover, since $\varphi|_{\tau_i}$ is satisfiable for each $i \in [2,n]$, the literals $\ell_{i,1}$ and $\ell_{i,2}$ are on distinct variables. Let $\alpha : \vec{x} \cup \vec{y} \to \{\bot,\top\}$ be a satisfying assignment for $F$. For each $i \in [2,n]$, \Cref{lem-two-edge-conflict} yields
    \begin{equation}\label{eq:falsified}
        F \land \ell_{i,1} \land \ell_{i,2} \models \bot.
    \end{equation}
    For each $i \in [2,n]$, let $\alpha_i$ be the assignment obtained from $\alpha$ by setting both $\ell_{i,1}$ and $\ell_{i,2}$ to $\top$. Then, by~\eqref{eq:falsified}, $\alpha_i$ falsifies $F$. Applying \Cref{lem-switching-witness} to the formula $F$, the satisfying assignment $\alpha$, and the assignments $\alpha_2,\dots,\alpha_n$, we obtain clauses $C_i \in F$ and literals $m_i$ such that $\alpha_i \models m_i$, $\alpha \models \overline{m_i}$, and $\overline{m_i} \in C_i$. Since $\alpha_i$ differs from $\alpha$ only on the variables $\mathrm{var}(\ell_{i,1})$ and $\mathrm{var}(\ell_{i,2})$, the literal $m_i$ must belong to $\{\ell_{i,1},\ell_{i,2}\}$. After possibly swapping the names of $\ell_{i,1}$ and $\ell_{i,2}$, we may assume that $m_i = \ell_{i,1}$.
    
    Now, $\{\ell_{i,1} \mid i \in [n]\}$ is a vertex cover of $G$, so $|\{\ell_{i,1} \mid i \in [2,n]\}| \ge \tau - 1$. We claim that, for all $i,j \in [2,n]$, if $\ell_{i,1} \neq \ell_{j,1}$, then $C_i \neq C_j$. Indeed, suppose that $\ell_{i,1} \neq \ell_{j,1}$ and $C_i = C_j$. Then, \Cref{lem-switching-witness} yields $\alpha_i \models \ell_{j,1}$ and $\alpha_j \models \ell_{i,1}$. Since $\alpha_i$ differs from $\alpha$ only on the variables $\mathrm{var}(\ell_{i,1})$ and $\mathrm{var}(\ell_{i,2})$, while $\alpha \models \overline{\ell_{j,1}}$, we obtain $\ell_{j,1} \in \{\ell_{i,1},\ell_{i,2}\}$. Similarly, $\ell_{i,1} \in \{\ell_{j,1},\ell_{j,2}\}$. Since $\ell_{i,1} \neq \ell_{j,1}$, this forces $\ell_{j,1} = \ell_{i,2}$ and $\ell_{i,1} = \ell_{j,2}$. Thus, $\{\ell_{i,1},\ell_{i,2}\} = \{\ell_{j,1},\ell_{j,2}\}$, contradicting \Cref{lem-distinct-edges}. Therefore, $|F| \ge \tau - 1$, so $|\psi| \ge |\psi|_{\tau_1}| \ge |F| - 2 \ge \tau - 3 \ge \sqrt{n+1} - 2$, as desired.
\end{proof}

\subsection{Finishing the proof}

\begin{lemma} \label{lem-lower-bound-p}
    If $\varphi(\vec{x},\vec{y})$ is a generalized $P$-encoding on $n$ input variables for some $n \ge 8$, then $|\varphi| \ge 2n + \sqrt{n+1} - 2$.
\end{lemma}
\begin{proof}
    Let $\mathcal{P}(n)$ be the minimum size of a generalized $P$-encoding on $n$ input variables. We may assume that $\varphi$ is prime and of size $\mathcal{P}(n)$. If there is a generalized $P$-encoding $\varphi'$ on $n$ input variables in regular form and such that $|\varphi'| = |\varphi|$, then we are done by \Cref{lem-regular-lb}. Otherwise, by \Cref{lem-regular}, there is a generalized $P$-encoding $\varphi'$ on $n-1$ input variables such that $|\varphi| \ge |\varphi'| + 3$. Hence, $\mathcal{P}(n) \ge \mathcal{P}(n-1)+3$ in this case. If $n=8$, then by \Cref{lem-weaker-lb}, we have $\mathcal{P}(8) \ge \mathcal{P}(7)+3 \ge 17 = 2\cdot 8 + \sqrt{8+1} - 2$. If $n \ge 9$, then we have $\mathcal{P}(n) \ge \mathcal{P}(n-1)+3 \ge 2(n-1) + \sqrt{n} + 1 \ge 2n + \sqrt{n+1} - 2$ by induction.
\end{proof}

\thmlowerbound*
\begin{proof}
    Every encoding of the $\amo(x_1, \ldots, x_n)$ constraint is a generalized $P$-encoding, so the result follows from \Cref{lem-lower-bound-p}.
\end{proof}

\section{\texorpdfstring{Proofs for \Cref{sec:gdpe}}{Proofs for disjunctive switching}}\label{appendix:dgpe}
The following is the key fact used to prove the correctness of the generalized product encoding, and we also use it in the proof of \Cref{thm:amk-disjunctive-product}. Given a tuple $\vec{i} \in [p]^k$ and $d \in [k]$, let $\vec{i}/d$ be $\vec{i}$ with its $d$th coordinate omitted.
\begin{restatable}{lemma}{lemproject} \label{lem-project}
    Given $k$ distinct points $\vec{i_1}, \dots, \vec{i_k} \in \mathbb{N}^k$, there is some $d \in [k]$ such that $\vec{i_1}/d, \dots, \vec{i_k}/d$ are distinct.\footnote{For points in $\{0,1\}^k$, this fact is known as Bondy's theorem~\cite{bondy}.}
\end{restatable}
\begin{proof}
    Suppose for a contradiction that we have $k$ distinct points $\vec{i_1}, \dots, \vec{i_k} \in \mathbb{N}^k$ and yet for each $d \in [k]$, we have $\vec{i_m}/d = \vec{i_n}/d$ for some distinct $m,n \in [k]$. We now create a graph whose vertices are the points $\vec{i_1}, \dots, \vec{i_k}$, and the edges are as follows. For each $d \in [k]$, choose distinct $m,n \in [k]$ such that $\vec{i_m}/d = \vec{i_n}/d$ and create an edge between $\vec{i_m}$ and $\vec{i_n}$. Note that this edge can be interpreted as saying that we can travel from $\vec{i_m}$ to $\vec{i_n}$ just by moving along one coordinate, which is determined by the edge.

    Since our graph has $k$ vertices and $k$ edges, it contains a cycle $C := (v_1, \ldots, v_c, v_1)$. Following the cycle $C$ geometrically means that we start from $v_1$, then move some nonzero amount in a sequence of orthogonal dimensions, and at the end, we are back to $v_1$; this is absurd.
\end{proof}

\amkdisjunctiveproduct*
\begin{proof}
    If $n \le (k+1)^k$, then use the parallel counter encoding from \cite{sinz}, which uses $O(n)$ clauses and auxiliary variables. Otherwise, rename the input variables $x_1,\dots,x_n$ to be of the form $x_{\vec{i}}$ with $\vec{i} \in [p]^{k+1}$, where $p = \ceil{n^{1/(k+1)}}$. Let $I \subseteq [p]^{k+1}$ be the set of tuples to which some variable is assigned. For each $d \in [k+1]$ and $\vec{i} \in [p]^k$, introduce an auxiliary variable $A_{d,\vec{i}}$. For each $d \in [2,k+1]$, introduce an auxiliary variable $w_d$. Then, the disjunctive generalized product encoding is given by the following constraints:
    \begin{framed}
    \begin{align*}
   &\bigwedge_{\vec{i} \in I} (\overline{x_{\vec{i}}} \lor A_{1,\vec{i}/1})  \tag{1} \\
    &\bigwedge_{\vec{i} \in I} \left(\overline{x_{\vec{i}}} \lor \bigvee_{d \in [2,k+1]} A_{d,\vec{i}/d}\right)  \tag{2}  \\
    &\bigwedge_{d \in [2,k+1]} \amk(\{A_{d,\vec{i}} \mid \vec{i} \in [p]^k\}) \tag{3} \\
    &\bigwedge_{d \in [2,k+1]} \bigwedge_{\vec{i} \in [p]^{k-1}} \left(\overline{w_d} \lor \amo\left(\left\{A_{1,\vec{i'}} \mid \vec{i'}/(d-1) = \vec{i}\right\}\right)\right)  \tag{4} \\
    &\amo(\{w_d \mid d \in [2,k+1]\})  \tag{5} \\
    &\bigwedge_{d \in [2,k+1]} \bigwedge_{\vec{i} \in [p]^k} \left(\overline{A_{d,\vec{i}}} \lor w_d\right). \tag{6}
    \end{align*}
    \end{framed}
    
    For the $\amk$ constraints within this encoding, we use the parallel counter encoding from \cite{sinz} (rather than recursion); for the $\amo$ constraints, we use any encoding with $O(n)$ clauses and auxiliary variables.

    First, we argue that the encoding is correct. Suppose that at most $k$ input variables are true. Let $A_{1,\vec{i}}$ be true if and only if $x_{\vec{i'}}$ is true for some $\vec{i'}$ with $\vec{i'}/1 = \vec{i}$. Then clauses (1) are satisfied. Let $I_1$ be the set of $\vec{i} \in [p]^k$ such that $A_{1,\vec{i}}$ is true. Clearly, $|I_1| \le k$. By \Cref{lem-project}, there is some $d \in [2,k+1]$ such that the $\vec{i}/(d-1)$ are distinct for $\vec{i} \in I_1$. Choose such a $d$ arbitrarily and let $w_d$ be true and the remaining $w_{d'}$ be false. Then, let $A_{d,\vec{i}}$ be true if and only if $x_{\vec{i'}}$ is true for some $\vec{i'}$ with $\vec{i'}/d = \vec{i}$, and let $A_{d',\vec{i}}$ be false for all $d' \in [2,k+1] \setminus \{d\}$. Then, clauses (2), (3), (5), and (6) are satisfied. Also, clauses (4) are satisfied by our choice of $d$. Thus, the formula is satisfiable.

    Conversely, suppose that at least $k+1$ input variables are true. By clauses (1), we must have $A_{1,\vec{i}/1}$ true for each $\vec{i} \in I$ such that $x_{\vec{i}}$ is true. By clauses (2) and (6), we must have $w_d$ true for some $d \in [2,k+1]$, so let $d$ be such that $w_d$ is true. By clauses (5), $w_{d'}$ is false for all $d' \in [2,k+1] \setminus \{d\}$. By clauses (2) and (6), we must have $A_{d,\vec{i}/d}$ true for each $\vec{i} \in I$ such that $x_{\vec{i}}$ is true. Then, by clauses (4), for each $\vec{i} \in [p]^{k-1}$, we have at most one $A_{1,\vec{i'}}$ true such that $\vec{i'}/(d-1) = \vec{i}$. Hence, for each $\vec{i} \in [p]^k$, we have at most one $x_{\vec{i'}}$ true such that $\vec{i'}/d = \vec{i}$. Thus, there are at least $k+1$ variables among $\{A_{d,\vec{i}} \mid \vec{i} \in [p]^k\}$ true, contradicting clauses (3). We conclude that the encoding is correct.

    Next, we count the number of clauses and auxiliary variables. The number of clauses in (1) and (2) is $n$ each; the number of clauses in (3), (4), and (6) is $O(k \cdot p^{k}) = O(kn^{k/(k+1)})$; the number of clauses in (5) is $O(k)$. In total, the number of clauses is $2n + O(kn^{k/(k+1)})$, as desired. The number of auxiliary variables of the form $A_{d,\vec{i}}$ is $(k+1) \cdot p^k = O(kn^{k/(k+1)})$; the number of auxiliary variables of the form $w_d$ is $k$; the number of auxiliary variables used by the $\amk$ and $\amo$ constraints in (3) and (4) is $O(kn^{k/(k+1)})$; the number of auxiliary variables used by the $\amo$ constraint in (5) is $O(k)$. In total, the number of auxiliary variables is $O(kn^{k/(k+1)})$, as desired.
\end{proof}

\section{\texorpdfstring{Proofs for \Cref{sec-grid-compression}}{Proofs for grid compression}}

Throughout this section, given a set family $\mathcal{H} := \{H_1, \ldots, H_m\}$ we will use notation $\mathcal{H}^{-1}(p) := \{j \in [m] \mid p \in H_j\}$ for the indices of set containing $p$.

\subsection{Grid compression encoding}
\lemhall*
\begin{proof}
    We show the existence of $\mathcal{H}$ by the probabilistic method. Suppose that each $H_i$ is chosen uniformly and i.i.d. from $\binom{[\ell]}{3}$. We will show that $\mathcal{H} = \{H_i \mid i \in [m]\}$ with nonzero probability admits a traversal. To do so, we use Hall's marriage theorem on each possible subset of size at most $k$. By Hall's marriage theorem~\cite{hall}, a subset $\mathcal{F} \subseteq \mathcal{H}$ admits a traversal if, for every subset $\mathcal{G} \subseteq \mathcal{F}$, 
    \[
    \left|\mathcal{G}\right| \leq \left|\bigcup \mathcal{G}\right|.
    \]
    Observe that $\{\mathcal{G}: \mathcal{G} \subseteq \mathcal{F} \subseteq \mathcal{H}, |\mathcal{F}| \leq k\} = \{ \mathcal{F} \subseteq \mathcal{H}: |\mathcal{F}| \leq k\}$, so it is sufficient to upper bound the probability that, for all $\mathcal{F} \subseteq \mathcal{H}$ with $|\mathcal{F}| \leq k$, 
    \[
    \left|\mathcal{F}\right| > \left|\bigcup \mathcal{F}\right|.
    \]

    Let $\mathcal{F} \subseteq \mathcal{H}$ be of size $\gamma$. Then, $\bigcup \mathcal{F}$ can be viewed as the result of sampling at least $3 \gamma$ elements with replacement, so we have
    \begin{align*}
        \prob{\left|\bigcup \mathcal{F}\right| < \gamma} &\le \sum_{\zeta \in \binom{[\ell]}{\gamma}} \prob{\bigcup \mathcal{F} \subseteq \zeta} \\
        &\le \binom{\ell}{\gamma} \lrp{\frac{\gamma}{\ell}}^{3\gamma}  \le \lrp{\frac{e \cdot \ell}{\gamma}}^\gamma \lrp{\frac{\gamma}{\ell}}^{3\gamma} \le e^\gamma \lrp{\frac{\gamma}{\ell}} ^ {2\gamma}.
    \end{align*}

    Having proved a bound for the probability that an individual subset $\mathcal{F}$ fails to satisfy the Hall condition, we can union bound the probability that we fail for some subset and thus get a lower bound on the probability that $\mathcal{H}$ exists with our desired properties:
    \begin{align*}
        \prob{\exists \mathcal{F} \in \binom{\mathcal{H}}{\leq k} : \size{\bigcup \mathcal{F}} < \size{\mathcal{F}}} &\leq \sum_{\gamma \in [k]} \sum_{\mathcal{F} \in \binom{\mathcal{H}}{\gamma}} \prob{\size{\bigcup \mathcal{F}} < \size{\mathcal{F}}} \\
        &\le \sum_{\gamma \in [k]} \binom{m}{\gamma} \cdot e^\gamma \lrp{\frac{\gamma}{\ell}} ^ {2\gamma} \\
        &\le \sum_{\gamma \in [k]} \lrp{\frac{e \cdot m}{\gamma}}^\gamma \cdot e^{\gamma} \cdot \lrp{\frac{\gamma}{\ell}} ^ {2\gamma} \\
        &\le \sum_{\gamma \in [k]} \lrp{\frac{e^2 m\gamma}{\ell^2}}^ {\gamma}.
    \end{align*}
    We can choose $m = \floor{\sqrt[3]{kn^2}}$ and $\ell = 4\ceil{\sqrt[3]{k^2n}}$, and since $\gamma \leq k$, we have
    \[
        \lrp{\frac{e^2 m\gamma}{\ell^2}} \leq \left( \frac{e^2 k^{1/3} n ^{2/3} \cdot k}{ 16 k^{4/3} n^{2/3}} \right) < \frac{1}{2},
    \]
    Therefore, 
    \[
     \prob{\exists \mathcal{F} \in \binom{\mathcal{H}}{\leq k} : \size{\bigcup \mathcal{F}} < \size{\mathcal{F}}} < \sum_{\gamma \in [k]} \left(\frac{1}{2}\right)^\gamma < \sum_{\gamma=1}^\infty \left(\frac{1}{2}\right)^\gamma  = 1,
    \]
    and thus the probability that $\mathcal{H}$ satisfies the desired properties is strictly positive.
\end{proof}

To prove \Cref{thm:conjunctive-compression}, we first prove that the grid compression encoding is unsatisfiable when more than $k$ variables are true:
\begin{lemma}
    The grid compression encoding for $\amk(x_1,\dots,x_n)$ is unsatisfiable when more than $k$ input variables are true. \label{lem:conjunctivegridcompression1}
\end{lemma}
\begin{proof}
    Suppose that at least $k+1$ input variables are set to true, and let $M_{i_1,j_1}, \dots, M_{i_{k+1},j_{k+1}}$ be distinct true input variables.
    By \eqref{eq:varimpliesitscolinM}, we have that $c_{j_\gamma}$ is true for every $\gamma \in [k+1]$. By \eqref{eq:setpossiblecopycolsconjunctive}, for each true $c_{j_\gamma}$, we have $\cp_{j_\gamma, p}$ for some $p \in H_{j_\gamma}$, and we denote by $p(\gamma)$ an arbitrary such $p$ chosen for $\gamma$. Then,
    $L_{i_\gamma,p(\gamma)}$ is true by \eqref{eq:colcopyconjunctive}.
    By \eqref{eq:columnmatching}, the map $j_\gamma \mapsto p(\gamma)$ is injective, and therefore the map $(i_\gamma,j_\gamma) \mapsto (i_\gamma,p(\gamma))$ is injective too. Thus, at least $k+1$ distinct variables $L_{i_\gamma,p(\gamma)}$ are set to true, contradicting \eqref{eq:atmostkinL}.
\end{proof}

\noindent
We can now prove \Cref{thm:conjunctive-compression}.

\thmgridcompression*
\begin{proof}
    First, we show that the grid compression encoding correctly encodes $\amk$. In light of \Cref{lem:conjunctivegridcompression1}, it suffices to show that the constraints are satisfiable when at most $k$ input variables are true.

    Suppose $M_{i_1,j_1}, \dots, M_{i_\gamma,j_\gamma}$ are the only true input variables for some $\gamma \leq k$. We wish to find an assignment that satisfies the constraints. We start by setting $c_{j}$ to true if and only if $j \in \{j_1, \dots, j_\gamma\}$; this satisfies \eqref{eq:varimpliesitscolinM}.

    By \Cref{lem:conjunctivegridcompression2}, we have that $\mathcal{F}:=\{H_{j_1}, \ldots, H_{j_\gamma}\}$ admits a transversal. Therefore, there are distinct $(p_1, \ldots, p_\gamma) \in H_{j_1} \times \ldots \times H_{j_{\gamma}}$. We set $\cp_{j_\eta, p_\eta}$ to be true for all $\eta \in [\gamma]$ and to be false for all other $\cp$ variables. Observe that we satisfy \eqref{eq:setpossiblecopycolsconjunctive}: for each $j\in[m]$, if $j = j_\eta$ for some $\eta \in [\gamma]$, then $\cp_{j_\eta, p_\eta}$ is true, and otherwise $c_j$ is false. Additionally, as $p_1, \ldots, p_\gamma$ are all distinct, we have that \eqref{eq:columnmatching} is satisfied: for a given $p \in [\ell]$, if $p = p_\eta$ for some $\eta \in [\gamma]$ then $\cp_{j_\eta, p_\eta}$ is the only true $\cp_{\_, p_\eta}$ variable; otherwise, all $\cp_{\_, p}$ variables are false.
    
    Finally, set $L_{i_\eta, p_\eta}$ to true, for each $\eta \in [\gamma]$, and set the other $L$ variables to false. As there are $\gamma$ pairs $(i_1,p_1), \ldots, (i_\gamma,p_\gamma)$ and $\gamma \leq k$, we automatically satisfy \eqref{eq:atmostkinL}. Lastly, we also satisfy \eqref{eq:colcopyconjunctive}: if $M_{i, j}$ and $\cp_{j, p}$, then $(i, j, p) = (i_\eta, j_\eta, p_\eta)$ for some $\eta \in [\gamma]$, and thus $L_{i, p}$ is set to true. Therefore, the encoding is correct.

    We next count the number of clauses in constraints \eqref{eq:atmostkinL}--\eqref{eq:columnmatching}:
    \begin{itemize}
        \item Constraint \eqref{eq:atmostkinL} consists of $O(n \ell / m) = O(\sqrt[3]{kn^2})$ clauses.
        \item Constraint \eqref{eq:varimpliesitscolinM} consists of $n$ clauses.
        \item Constraint \eqref{eq:colcopyconjunctive} consists of $3n$ clauses.
        \item Constraint \eqref{eq:setpossiblecopycolsconjunctive} consists of $m = O(\sqrt[3]{kn^2})$ clauses.
        \item Constraint \eqref{eq:columnmatching}, using a linear-size encoding for $\amo$ (e.g.,~\Cref{thm:am1}), consists of
        \[
            \sum_{p \in [\ell]} O(|\mathcal{H}^{-1}(p)|) = O\left(\sum_{j \in [m]} |H_j|\right) = O(m) = O(\sqrt[3]{kn^2})
        \]
        clauses.
   \end{itemize}
    Thus, the total number of clauses is $4n + O(\sqrt[3]{kn^2})$, as desired.

    Finally, we count the number of auxiliary variables:
    \begin{itemize}
        \item We have $O(n \ell / m) = O(\sqrt[3]{kn^2})$ auxiliary variables of the form $L_{i,p}$.
        \item We have $m = O(\sqrt[3]{kn^2})$ auxiliary variables of the form $c_j$.
        \item We have $3m = O(\sqrt[3]{kn^2})$ auxiliary variables of the form $\cp_{j,p}$.
    \end{itemize}
    Thus, the total number of auxiliary variables is $O(\sqrt[3]{kn^2})$, as desired.
\end{proof}

\subsection{Disjunctive grid compression encoding}

\lemreedsolomon*
\begin{proof}
    Choose $q = \Theta(k \log_k n)$ to be a prime power, e.g., $q = 2^r$ for a suitable $r$. We describe how to construct a family of sets $\mathcal{H}$ such that $H_i \in \binom{\mathbb{F}_{q}^2}{q}$ for each $H_i \in \mathcal{H}$; then we can rename the elements of $\mathbb{F}_{q}^2$ to be elements of $[\ell]$, where $\ell = q^2 = \Theta\lrp{k^2\log^2_kn}$. Given a polynomial $f \in \mathbb{F}_{q}[x]$, let $S_f = \{(x, f(x)) : x \in \mathbb{F}_{q}\} \in \binom{\mathbb{F}_{q}^2}{q}$. Then, let
    \[
        \mathcal{H} = \left\{S_f : f \in \mathbb{F}_{q}[x],\; \deg(f) < \ceil{\frac{q}{k-1}} \right\}.
    \]
    Note that a polynomial of degree strictly less than $d$ is determined by the choice of $d$ coefficients, so $|\mathcal{H}| = q^{\ceil{q/(k-1)}} = \Theta(k \log_k n)^{\Theta(\log_k n)} \ge n = \omega(\sqrt{nk\log_kn})$ for suitable choices of constants.
    
    We claim that $\mathcal{H}$ is $(k-1)$-cover-free. Indeed, if $S_{f_k} \subseteq S_{f_1} \cup S_{f_2} \cup \dots \cup S_{f_{k-1}}$ for distinct $S_{f_1},\dots,S_{f_k} \in \mathcal{H}$, then by the pigeonhole principle, $|S_{f_k} \cap S_{f_j}| \ge \ceil{q/(k-1)}$ for some $j \in [k-1]$. This implies that $f_k - f_j$ has at least $\ceil{q/(k-1)}$ roots in $\mathbb{F}_{q}$, a contradiction.
\end{proof}

To prove \Cref{thm:disjunctive-compression}, we first prove that the disjunctive grid compression encoding is unsatisfiable when more than $k$ variables are true:
\begin{lemma} \label{lem-dis-grid-unsat}
    The disjunctive grid compression encoding for $\amk(x_1,\dots,x_n)$ is unsatisfiable when more than $k$ input variables are true. \label{lem:disjointgridcompression1}
\end{lemma}
\begin{proof}
    Suppose that at least $k+1$ input variables are set to true, and let $M_{i_1,j_1}, \dots, M_{i_{k+1},j_{k+1}}$ be distinct true input variables. By \eqref{eq:varimpliesitscolinM}, we have that $c_{j_\gamma}$ is true for every $\gamma \in [k+1]$. By \eqref{eq:colcopydisjunctive}, we have $L_{i_\gamma,p}$ for some $p \in H_{j_\gamma}$ for every $\gamma \in [k+1]$, and we denote by $p(\gamma)$ an arbitrary such $p$ chosen for $\gamma$. By \eqref{eq:removeoverloadedcols}, we have $\overline{\ov_{p(\gamma)}}$ for each $\gamma \in [k+1]$. Then, in turn, \eqref{eq:getoverloadedcols} implies that, for each $\gamma \in [k+1]$, we have at most one $c_j$ true such that $p(\gamma) \in H_j$. Thus, the map $(i_\gamma,j_\gamma) \mapsto (i_\gamma,p(\gamma))$ is injective. Thus, at least $k+1$ distinct variables $L_{i_\gamma,p(\gamma)}$ are set to true, contradicting \eqref{eq:atmostkinL}.
\end{proof}

Next, we prove that the disjunctive grid compression encoding is satisfiable when at most $k$ variables are true:

\begin{lemma} \label{lem-dis-grid-sat}
    The disjunctive grid compression encoding for $\amk(x_1,\dots,x_n)$ is satisfiable if $k$ or fewer of the input variables are true.
\label{lem:disjointgridcompression2}
\end{lemma}

\begin{proof}
    Suppose $M_{i_1,j_1}, \dots, M_{i_\gamma,j_\gamma}$ are the only true input variables for some $\gamma \leq k$. We wish to find an assignment that satisfies the constraints. We start by setting $c_{j}$ to true if and only if $j \in \{j_1, \dots, j_\gamma\}$; this satisfies \eqref{eq:varimpliesitscolinM}.

    By \Cref{lem:disjointgridcompression3}, no set from $\{H_{j_1},\dots,H_{j_\gamma}\}$ is covered by the remaining sets, of which there are at most $k-1$. Thus, for any $H_{j_\eta} \in \{H_{j_1},\dots,H_{j_\gamma}\}$ there exists some element $p \in H_{j_\eta}$ such that
    \begin{equation} \label{eq:cover-free}
        p \notin \bigcup_{\substack{\eta' \in [\gamma] \\ j_{\eta'} \neq j_\eta}} H_{j_{\eta'}},
    \end{equation}
    and we denote by $p(j_\eta)$ an arbitrary such $p$ chosen for $\gamma$. Therefore, $j_\eta \mapsto p(j_\eta)$ is an injection from $\{j_1,\dots,j_\gamma\}$ to $[\ell]$ such that $p(j_\eta) \in H_{j_\eta}$ for each $\eta \in [\gamma]$. We set $L_{i_\eta,p(j_\eta)}$ to true for all $\eta \in [\gamma]$ and all other $L$ variables to false. Observe that we satisfy \eqref{eq:atmostkinL} and \eqref{eq:colcopydisjunctive}. Also, set $\ov_{p(j_\eta)}$ to false for each $\eta \in [\gamma]$, and set all other $\ov$ variables to true. This satisfies \eqref{eq:removeoverloadedcols} by construction.
    For~\eqref{eq:getoverloadedcols}, it suffices to consider the cases in which $\ov_p$ is set to false, which by construction means $p = p(j_\eta)$ for some $\eta \in [\gamma]$. Then, \eqref{eq:cover-free} implies that $\{j_1,\dots,j_\gamma\} \cap \mathcal{H}^{-1}(p(j_\eta)) = \{j_\eta\}$, so $c_{j_\eta}$ is the only true variable in $\{c_j \mid j \in \mathcal{H}^{-1}(p)\}$. Thus, \eqref{eq:getoverloadedcols} is satisfied.
\end{proof}

\noindent
We can now prove~\Cref{thm:disjunctive-compression}.

\thmdisjunctivegridcompression*
\begin{proof}
    The correctness of the encoding was proven in \Cref{lem-dis-grid-unsat,lem-dis-grid-sat}. It remains to count the number of clauses and auxiliary variables. First, we count the number of clauses:
    \begin{itemize}
        \item Constraint \eqref{eq:atmostkinL} consists of $O(n \ell / m) = O\lrp{\sqrt{nk^3\log^3_k n}}$ clauses.
        \item Constraint \eqref{eq:varimpliesitscolinM} consists of $n$ clauses.
        \item Constraint \eqref{eq:colcopydisjunctive} consists of $n$ clauses.
        \item Constraint \eqref{eq:getoverloadedcols} consists of $O\lrp{\sum_{j \in [m]} |H_j|} = O(m k \log_k n) = O\lrp{\sqrt{nk^3\log^3_k n}}$ clauses.
        \item Constraint \eqref{eq:removeoverloadedcols} consists of $O(n \ell / m) = O\lrp{\sqrt{nk^3\log^3_k n}}$ clauses.
    \end{itemize}
    Thus, the total number of clauses is $2n + O\lrp{\sqrt{nk^3\log^3_k n}}$, as desired

    Finally, we count the number of auxiliary variables:
    \begin{itemize}
        \item We have $O(n \ell / m) = O\lrp{\sqrt{nk^3\log^3_k n}}$ auxiliary variables of the form $L_{i,p}$.
        \item We have $m = O(\sqrt{nk\log_kn})$ auxiliary variables of the form $c_j$.
        \item We have $\ell = O(k^2\log^2_kn)$ auxiliary variables of the form $\ov_p$.
    \end{itemize}
    Thus, the total number of auxiliary variables is $O\lrp{\sqrt{nk^3\log^3_k n}}$, as desired.
\end{proof}

\section{\texorpdfstring{Proof of \Cref{thm:am1-clique}}{Proof of the clique encoding theorem}}
\thmamoclique*
\begin{proof}
    Let $K_p$ be the complete graph on $p := \ceil{\sqrt{2n}}+1$ vertices. Then, $E(K_p) = \binom{p}{2} \ge n$. Assign the variables $x_1,\dots,x_n$ to distinct edges of $G$, renaming the variables so that $x_{\{i,j\}}$ is the variable assigned to the edge $\{i,j\}$. Let $E$ be the set of edges of $G$ to which some variable is assigned. Discard any vertices of $G$ not incident to an edge from $E$. Introduce auxiliary variables $v_i$ for each $i \in V(G)$. Our encoding is then as follows:
    \[
        \textsf{CE}(\{x_{\{i,j\}} \mid \{i,j\} \in E\}) := \left(\bigwedge_{\{i,j\} \in E} (\overline{x_{\{i, j\}}} \lor v_i) \land (\overline{x_{\{i, j\}}} \lor v_j) \right) \land \amt(v_1,\dots,v_p),
    \]
    where we use the disjunctive grid compression encoding for $\amt(v_1,\dots,v_p)$, which uses $2p + \widetilde{O}(\sqrt{p})$ clauses and $\widetilde{O}(\sqrt{p})$ auxiliary variables. The total number of clauses is $2n + 2p + \widetilde{O}(\sqrt{p}) = 2n + 2 \sqrt{2n} + \widetilde{O}(\sqrt[4]{n})$, and the total number of auxiliary variables is $p + \widetilde{O}(\sqrt{p}) = \sqrt{2n} + \widetilde{O}(\sqrt[4]{n})$.

    It remains to argue that the encoding is correct. Suppose that at most one input variable is true. If zero input variables are true, then $\textsf{CE}(\{x_{\{i,j\}} \mid \{i,j\} \in E\})$ is satisfiable by setting all of the $v$ auxiliary variables to false. Otherwise, exactly one input variable $x_{\{i,j\}}$ is true. Then, $\textsf{CE}(\{x_{\{i,j\}} \mid \{i,j\} \in E\})$ is satisfiable by setting $v_i$ and $v_j$ to true and all of the other $v$ auxiliary variables to false.

    On the other hand, suppose that at least two input variables $x_{\{i,j\}}$ and $x_{\{i',j'\}}$ are true. Then, $v_i$, $v_j$, $v_{i'}$, and $v_{j'}$ must be true. Since $|\{i,j,i',j'\}| \ge 3$, this falsifies the $\amt(v_1,\dots,v_p)$ constraint.
\end{proof}


\section{Disjunctive Switching}\label{app:disjunctive-switching}
Consider the following excerpt of pseudocode, representing an algorithm we wish to encode in CNF:
\begin{pythonPrime}
# Assume functions f1: [n] -> Y1, f2: [n] -> Y2, f3: [n] -> Y3
# with Y1, Y2, and Y3 pairwise disjoint sets.
for i in [n]:
  if x[i]:
    switch state:
        case c1:
            y[f1(i)] = True
        case c2:
            y[f2(i)] = True
        case c3:
            y[f3(i)] = True
\end{pythonPrime}
Using variables $x_i$ and $y_k$, as well as $c_1, c_2, c_3$ to represent the cases, the direct encoding would be to simply enforce:
\begin{equation}\label{eq:pre-disj}
     \bigwedge_{i=1}^n \bigwedge_{j=1}^3 \lrp{(x_i \land c_j) \rightarrow y_{f_j(i)}}.
\end{equation}
Suppose that, in this context, we may assume that exactly one $c_j$ variable is true (i.e., the conditions are exhaustive and mutually exclusive) and that $y_{f_j(i)}$ is set to true only if $c_j$ is true for each $i \in [n]$ and $j \in [3]$ (e.g., \pythin{y[f2(i)]} can only be assigned  \pythin{True} under the  \pythin{c2} branch; no other part of the code can make it true). Then, the idea of disjunctive switching is that we may replace the $3n$ clauses from~\eqref{eq:pre-disj} with the $n + |Y_1| + |Y_2| + |Y_3|$ clauses
\begin{equation}\label{eq:disj-switching-2}
        \bigwedge_{i \in [n]} \lrp{x_i \rightarrow \bigvee_{j \in [3]} y_{f_j(i)}} \qquad \text{and} \qquad \bigwedge_{j \in [3]} \bigwedge_{k \in Y_j} \lrp{y_k \rightarrow c_j}.
\end{equation}
At a high level, the set of clauses on the left enforces that each active $x_i$ will activate at least one of the three possible $y$ variables, whereas the set of clauses on the right enforces that the activated $y$ variables are consistent with the unique condition $c_j$ that holds. If $|Y_1| + |Y_2| + |Y_3| < 2n$, then this transformation shrinks the encoding. Moreover, we have chosen 3 conditions to simplify our example, but the technique becomes more powerful when the number of conditions $C$ is large; in our applications, we will have $\sum_{j=1}^C |Y_j| = o(n)$, and thus achieve a $C \cdot n \to n + o(n)$ clause reduction.

Making an encoding amenable to disjunctive switching may require massaging it in a subtle way, so we regard disjunctive switching as a paradigm for constructing encodings rather than a formal transformation.

\section{Empirical Evaluation}\label{sec:empirical}

This section presents a \emph{preliminary} evaluation of our encodings in practice. We begin by stressing that an exhaustive evaluation of cardinality constraints encodings, including diverse families of instances, different solvers, and solver parameters, is a challenging task on its own, with full experimental papers dedicated to it~\cite{empirical-study, wynn2018comparisonencodingscardinalityconstraints, bittnerSATEncodingsAtMostk2019}. Our goal, therefore, is more modest: to obtain empirical evidence of whether our encodings for the $\amk$ constraint can be practical, at least in some cases, despite their lack of propagation completeness. 
We omit presenting experiments regarding $\amo$ encodings, since our improvements take place in the second-order term, and for most practical instances the number of clauses is almost the same as from the product encoding, with almost the same performance. We focus on the case $k=2$, since our goal is simply to perform an initial exploration 
of the landscape.

\subparagraph*{Instance Descriptions.} We experimented with the following families of instances:
\begin{itemize}
    \item (\textbf{Family L}) Given $n \geq 10k$ and $k \geq 2$, there are variables $x_1, \dots, x_n$ and we impose a single $\amk(x_1, \dots, x_n)$ constraint. To make the instance unsatisfiable, $k+1$ disjoint subsets $S_1, \dots, S_{k+1} \subseteq \{1, \dots, n\}$, with $|S_i| = 10$ for every $i \in [k+1]$, are chosen randomly, and we add one clause $\bigvee_{j \in S_i} x_j$ for each $S_i$. Concretely, we let $U_0 := \{1, \dots, n\}$, and then for $i \in [k+1]$, we sample a uniformly random size-10 subset $S_i$ of $U_{i-1}$, and then set $U_i := U_{i-1} \setminus S_i$. We also consider satisfiable instances by choosing $k$ disjoint subsets instead of $k+1$. Instances from this family are solvable for very large values of $n$ in both the SAT and UNSAT case, hence the name of the family.
    \item (\textbf{Family M}) This family of instances is reminiscent of pigeonhole principle formulas, yet designed for having more indirection. There are $M$ machines, $T$ jobs, and a ``capacity''  parameter $c$. Consider variables $x_{m, t}$, for $1\leq m \leq M$ and $1 \leq t \leq T$ corresponding to whether job $t$ runs on machine $m$. Then, for each job $t$, we add clause $\bigvee_{m=1}^M x_{m, t}$ to enforce that it runs on some machine. Next, for each machine $m$, a variable $a_m$ represents that the machine is \emph{active}, and only active machines can run jobs, as enforced by clauses $\overline{x_{m, t}} \lor a_m$ for every $1 \leq t \leq T$. Finally, we enforce that at most $k$ machines are active, and that each machine receives at most $c$ jobs, which makes the instance unsatisfiable if we choose $T = k \cdot c + 1$ and satisfiable if $T = k \cdot c$.
    \item (\textbf{Family D}) Given $n$, we build a ``layered'' acyclic graph as follows: layer $1$ contains a single vertex $s$, layers $2$ through $k$ contain $n$ vertices each, and layer $k+1$ contains a single vertex $t$. Each vertex in layer $i$ has a directed edge toward each vertex in layer $i+1$, for $1 \leq i \leq k$. To each vertex $v$ we associate a variable $x_v$ representing that the vertex is ``active'', and we enforce that for each active vertex $v$ other than $t$, if $v$ is active then at least one of its out-neighbors is active. To make the instance unsatisfiable, we enforce that $s$ is active, and at most $k$ vertices are active.
\end{itemize}

\subparagraph*{Hardware and Software.} We ran all experiments on a MacBook Pro personal computer, with 16 GB of RAM, an Apple M5 chip, and running macOS Tahoe 26.2; all experiments were single-threaded. We used the PySAT library~\cite{imms-sat18} in its version \texttt{python-sat 1.8.dev27}, together with its included solver \textsf{CaDiCaL} version 1.9.5. For validation purposes, since the PySAT version is by now rather old, we also solved instances with~\textsf{kissat} version 4.0.4~\cite{BiereFallerFazekasFleuryFroleyksPollitt-SAT-Competition-2024-solvers}, and found the runtime trends to be similar to those of~\textsf{CaDiCaL}. Hence, our code uses PySAT's \textsf{CaDiCaL} to make our experiments directly runnable.

\subparagraph*{Encodings and Reproducibility.} We experimented with the different $\amk$ encodings available in PySAT (see \url{https://pysathq.github.io/docs/html/api/card.html#pysat.card.CardEnc.atmost}). Out of these, the ones that performed best are the sequential counter (\texttt{seqcounter}) encoding, the cardinality network encoding (\texttt{cardnetwrk}), and the modulo totalizer encoding (\texttt{mtotalizer}). Since PySAT's implementation of the sequential counter is slow to generate the encoding, we re-implemented this encoding generator from scratch.
To avoid cluttering, we sometimes present only the best-performing encoding from the PySAT suite, as in~\Cref{sec:discussion}. Outside of the PySAT suite of encodings, we include experiments with:
\begin{itemize}
    \item The generalized product encoding (GP) of Frisch and Giannaros, as described in~\Cref{sec:gdpe}.
    \item The disjunctive generalized product encoding (DGP).
    \item The disjunctive grid compression encoding (DGC). 
\end{itemize} 
With respect to the DGC encoding, for the particular case of $k = 2$ we include an optimization in the construction of the set collection $\mathcal{H}$; we simply take distinct sets of size $2$, which is enough since a $(k-1)$-cover-free family in this case is a $1$-cover-free family, which is to say, a family of sets such that none contains another.\footnote{A family of sets with this property is called a \emph{Sperner family}~\cite{sperner}.}
Moreover, the particular choice of the parameters $m$ and $\ell$ is done by a grid-search procedure aiming to minimize the number of clauses. 

While the GP encoding is implemented recursively, for the DGP and DGC encodings we use the sequential counter encoding for the recursive $\amk$ constraints. We experimented with deeper recursions without seeing any significant advantage in performance.

\textbf{Our code:} \url{https://github.com/bsubercaseaux/cardinality-encodings/}

\subsection{Results}

\Cref{tab:atmost2-clause-counts,tab:atmost2-aux-counts} report the encoding-size data $k=2$ over large $n$ values, comparing DGC, DGP, GP, and the sequential counter baseline.   \Cref{tab:family-L-unsat,tab:family-L-sat,tab:family-M-sat,tab:family-M-unsat,tab:family-D} present runtime data for the different families, both in SAT and UNSAT regimes when applicable. As can be observed in~\Cref{tab:family-M-sat,tab:family-M-unsat}, the encodings lacking propagation completeness perform significantly worse for Family M, most especially in the UNSAT case; this family is a variant of the pigeonhole principle, and its UNSAT case seems to require an exhaustive exploration of the search space. Our experiments suggest that in these cases, propagation completeness might indeed be necessary for efficiency. On the other hand, per~\Cref{tab:family-L-unsat,tab:family-L-sat} our encodings perform very well on Family L; a relatively easy family where $n$ gets to be very large. Finally, as shown in~\Cref{tab:family-D}, Family D represents a more mixed case, in which our encodings perform slightly worse than the best encoding from PySAT yet slightly better than other PySAT encodings.

\begin{table}[t]
\centering
\caption{Clause counts for $\amt$ encodings}
\label{tab:atmost2-clause-counts}
\small
\setlength{\tabcolsep}{5pt}
\begin{tabular}{r|rrrr}
\toprule
$n$ & Seq. Counter & GP & DGP & DGC \\
\midrule
200,000 & 999,993 & 654,117 & 462,163 & 448,996 \\
400,000 & 1,999,993 & 1,273,908 & 897,953 & 877,578 \\
600,000 & 2,999,993 & 1,892,916 & 1,329,347 & 1,301,906 \\
800,000 & 3,999,993 & 2,506,839 & 1,754,915 & 1,723,022 \\
1,000,000 & 4,999,993 & 3,120,159 & 2,179,177 & 2,143,170 \\
1,200,000 & 5,999,993 & 3,737,979 & 2,605,203 & 2,561,360 \\
1,400,000 & 6,999,993 & 4,349,103 & 3,024,873 & 2,979,142 \\
1,600,000 & 7,999,993 & 4,959,408 & 3,445,443 & 3,395,660 \\
1,800,000 & 8,999,993 & 5,571,486 & 3,866,913 & 3,811,884 \\
2,000,000 & 9,999,993 & 6,181,791 & 4,284,737 & 4,227,056 \\
2,200,000 & 10,999,993 & 6,793,356 & 4,707,827 & 4,641,736 \\
2,400,000 & 11,999,993 & 7,401,942 & 5,122,113 & 5,056,372 \\
2,600,000 & 12,999,993 & 8,011,734 & 5,541,665 & 5,470,200 \\
2,800,000 & 13,999,993 & 8,622,291 & 5,956,707 & 5,883,808 \\
3,000,000 & 14,999,993 & 9,232,587 & 6,377,267 & 6,297,534 \\
\bottomrule
\end{tabular}
\end{table}

\begin{table}[t]
\centering
\caption{Auxiliary-variable counts for $\amt$ encodings}
\label{tab:atmost2-aux-counts}
\small
\setlength{\tabcolsep}{5pt}
\begin{tabular}{r|rrrr}
\toprule
$n$ & Seq. Counter & GP & DGP & DGC \\
\midrule
200,000 & 399,998 & 20,955 & 31,205 & 24,656 \\
400,000 & 799,998 & 27,552 & 49,130 & 38,981 \\
600,000 & 1,199,998 & 33,888 & 64,849 & 51,173 \\
800,000 & 1,599,998 & 38,529 & 77,649 & 61,751 \\
1,000,000 & 1,999,998 & 42,969 & 89,794 & 71,837 \\
1,200,000 & 2,399,998 & 48,909 & 102,821 & 80,954 \\
1,400,000 & 2,799,998 & 52,617 & 112,666 & 89,867 \\
1,600,000 & 3,199,998 & 56,052 & 122,961 & 98,132 \\
1,800,000 & 3,599,998 & 60,078 & 133,706 & 106,250 \\
2,000,000 & 3,999,998 & 63,513 & 142,626 & 113,840 \\
2,200,000 & 4,399,998 & 67,368 & 154,181 & 121,202 \\
2,400,000 & 4,799,998 & 70,230 & 161,330 & 128,534 \\
2,600,000 & 5,199,998 & 73,494 & 171,114 & 135,446 \\
2,800,000 & 5,599,998 & 77,013 & 178,641 & 142,262 \\
3,000,000 & 5,999,998 & 80,445 & 188,929 & 149,129 \\
\bottomrule
\end{tabular}
\end{table}

\begin{table}[htbp]
\centering
\caption{Runtime comparison for Family L (UNSAT). Times are in milliseconds.}
\begin{tabular}{r | r r r r}
\hline
    $n$ & Seq. Counter & GP & DGP & DGC \\
\hline
    100,000 & 21 & 160 & 21 & 12 \\
    200,000 & 40 & 396 & 49 & 28 \\
    300,000 & 63 & 746 & 77 & 53 \\
    400,000 & 81 & 1,090 & 105 & 55 \\
    500,000 & 101 & 1,622 & 135 & 71 \\
    600,000 & 139 & 1,929 & 190 & 97 \\
    700,000 & 146 & 2,504 & 245 & 102 \\
    800,000 & 164 & 2,550 & 287 & 121 \\
    900,000 & 200 & 2,691 & 307 & 115 \\
    1,000,000 & 218 & 3,171 & 252 & 134 \\
    1,100,000 & 232 & 3,939 & 326 & 170 \\
    1,200,000 & 262 & 4,185 & 358 & 169 \\
    1,300,000 & 277 & 3,816 & 438 & 170 \\
    1,400,000 & 328 & 5,658 & 618 & 227 \\
    1,500,000 & 355 & 5,163 & 483 & 252 \\
    1,600,000 & 371 & 5,880 & 598 & 239 \\
    1,700,000 & 410 & 6,570 & 641 & 257 \\
    1,800,000 & 445 & 6,209 & 566 & 285 \\
    1,900,000 & 474 & 6,756 & 582 & 240 \\
    2,000,000 & 501 & 6,503 & 571 & 254 \\
    2,100,000 & 539 & 8,227 & 623 & 323 \\
    2,200,000 & 573 & 7,015 & 737 & 337 \\
    2,300,000 & 610 & 9,410 & 769 & 334 \\
    2,400,000 & 625 & 8,658 & 798 & 346 \\
    2,500,000 & 677 & 9,856 & 826 & 378 \\
    2,600,000 & 803 & 9,960 & 856 & 387 \\
    2,700,000 & 778 & 10,839 & 876 & 419 \\
    2,800,000 & 856 & 12,628 & 976 & 503 \\
    2,900,000 & 850 & 12,601 & 945 & 401 \\
    3,000,000 & 848 & 10,013 & 1,051 & 434 \\
\hline
\end{tabular}
\label{tab:family-L-unsat}
\end{table}

\begin{table}[htbp]
\centering
\caption{Runtime comparison for Family L (SAT). Times are in milliseconds.}
\begin{tabular}{r | r r r r}
\hline
    $n$ & Seq. Counter & GP & DGP & DGC \\
\hline
    100,000 & 7 & 4 & 3 & 3 \\
    200,000 & 17 & 9 & 7 & 7 \\
    300,000 & 24 & 15 & 11 & 10 \\
    400,000 & 32 & 22 & 15 & 13 \\
    500,000 & 43 & 29 & 20 & 17 \\
    600,000 & 50 & 34 & 24 & 20 \\
    700,000 & 62 & 40 & 53 & 22 \\
    800,000 & 76 & 46 & 31 & 25 \\
    900,000 & 82 & 51 & 34 & 28 \\
    1,000,000 & 88 & 59 & 71 & 32 \\
    1,100,000 & 119 & 66 & 47 & 43 \\
    1,200,000 & 124 & 75 & 47 & 46 \\
    1,300,000 & 121 & 90 & 50 & 47 \\
    1,400,000 & 170 & 86 & 52 & 52 \\
    1,500,000 & 193 & 94 & 59 & 57 \\
    1,600,000 & 214 & 99 & 61 & 55 \\
    1,700,000 & 228 & 136 & 67 & 61 \\
    1,800,000 & 227 & 114 & 84 & 65 \\
    1,900,000 & 244 & 119 & 75 & 159 \\
    2,000,000 & 284 & 162 & 79 & 92 \\
    2,100,000 & 268 & 136 & 88 & 74 \\
    2,200,000 & 272 & 144 & 91 & 78 \\
    2,300,000 & 305 & 154 & 96 & 83 \\
    2,400,000 & 331 & 156 & 114 & 85 \\
    2,500,000 & 346 & 183 & 112 & 107 \\
    2,600,000 & 356 & 176 & 126 & 100 \\
    2,700,000 & 386 & 179 & 112 & 94 \\
    2,800,000 & 444 & 204 & 122 & 99 \\
    2,900,000 & 464 & 201 & 136 & 98 \\
    3,000,000 & 457 & 214 & 131 & 106 \\
\hline
\end{tabular}
\label{tab:family-L-sat}
\end{table}

\begin{table}[htbp]
\centering
\caption{Runtime comparison for Family M (UNSAT). Times are in milliseconds.}
\begin{tabular}{r | r r r r r r}
\hline
    $n$ & Seq. Counter & GP & DGP & DGC & cardnetwrk & mtotalizer \\
\hline
    20 & 12 & 19 & 17 & 17 & 19 & 12 \\
    30 & 32 & 46 & 52 & 44 & 49 & 33 \\
    40 & 94 & 86 & 103 & 96 & 85 & 76 \\
    50 & 134 & 174 & 217 & 149 & 137 & 128 \\
    60 & 240 & 243 & 417 & 398 & 248 & 231 \\
    70 & 308 & 423 & 730 & 522 & 408 & 347 \\
    80 & 510 & 731 & 1,207 & 1,137 & 516 & 500 \\
    90 & 775 & 886 & 3,841 & 1,213 & 764 & 714 \\
    100 & 1,073 & 1,201 & 4,242 & 10,474 & 1,095 & 944 \\
\hline
\end{tabular}
\label{tab:family-M-unsat}
\end{table}

\begin{table}[htbp]
\centering
\caption{Runtime comparison for Family M (SAT). Times are in milliseconds.}
\begin{tabular}{r | r r r r r r}
\hline
    $n$ & Seq. Counter & GP & DGP & DGC & cardnetwrk & mtotalizer \\
\hline
    20 & 0 & 0 & 0 & 0 & 0 & 0 \\
    1,020 & 0 & 0 & 1 & 3 & 0 & 0 \\
    2,020 & 0 & 0 & 2 & 10 & 1 & 0 \\
    3,020 & 0 & 0 & 4 & 25 & 1 & 34 \\
    4,020 & 1 & 0 & 4 & 37 & 2 & 60 \\
    5,020 & 1 & 1 & 6 & 56 & 3 & 85 \\
    6,020 & 1 & 1 & 8 & 82 & 3 & 2 \\
    7,020 & 2 & 1 & 8 & 114 & 4 & 160 \\
    8,020 & 3 & 1 & 12 & 136 & 5 & 3 \\
    9,020 & 2 & 2 & 12 & 3,240 & 5 & 4 \\
    10,020 & 3 & 2 & 12 & 3,976 & 6 & 455 \\
    11,020 & 4 & 2 & 17 & 4,559 & 9 & 6 \\
    12,020 & 5 & 3 & 21 & 5,489 & 9 & 6 \\
    13,020 & 5 & 3 & 19 & 6,363 & 10 & 8 \\
    14,020 & 5 & 3 & 19 & 7,175 & 10 & 624 \\
\hline
\end{tabular}
\label{tab:family-M-sat}
\end{table}

\begin{table}[htbp]
\centering
\caption{Runtime comparison for Family D (UNSAT). Times are in milliseconds.}
\begin{tabular}{r | r r r r r}
\hline
    $n$ & Seq. Counter & GP & DGP & DGC & cardnetwrk \\
\hline
    500 & 14 & 76 & 88 & 51 & 113 \\
    1,000 & 83 & 337 & 409 & 211 & 132 \\
    1,500 & 348 & 718 & 792 & 1,258 & 1,221 \\
    2,000 & 1,098 & 1,551 & 2,729 & 2,847 & 845 \\
    2,500 & 2,254 & 3,469 & 5,340 & 2,157 & 6,497 \\
    3,000 & 3,992 & 7,102 & 12,451 & 5,093 & 7,950 \\
\hline
\end{tabular}
\label{tab:family-D}
\end{table}

\end{document}